\documentclass[11pt]{article}

\usepackage[letterpaper,margin=1in]{geometry}

\usepackage{amsmath, amsthm, amssymb}

\usepackage{url}

\usepackage{times}

\newtheorem{theorem}{Theorem}[section]
\newtheorem{lemma}[theorem]{Lemma}

\newtheorem{corollary}[theorem]{Corollary}

\newtheorem{definition}[theorem]{Definition}

\def\paragraph#1{\vspace{0.05em}\noindent {\bf #1}}
\def\paragraph#1{\vspace{0.25em}\noindent {\bf #1}}

\begin{document}

\title{Gossip in a Smartphone Peer-to-Peer Network}
          
\author{Calvin Newport\\ Georgetown University\\ Washington, DC\\ {\tt cnewport@cs.georgetown.edu}}
\date{}
\maketitle

\begin{abstract}
In this paper, we study the fundamental problem of gossip in the {\em mobile telephone model}:
a recently introduced variation of the classical {\em telephone model} modified to better describe the local peer-to-peer
communication services implemented in many popular smartphone operating systems.
In more detail, the mobile telephone model differs from the classical telephone model in three ways: (1) each device
can participate in at most one connection per round; (2) the network topology can undergo a parameterized
rate of change; and (3) devices can advertise a parameterized number of bits about their state to their neighbors in each round before connection attempts are initiated.
We begin by describing and analyzing new randomized gossip algorithms in this model under the harsh assumption of a network topology
that can change completely in every round.
We prove a significant time complexity gap between the case where nodes can advertise $0$ bits to their neighbors in each round,
and the case where nodes can advertise $1$ bit.
For the latter assumption, we present two solutions: the first depends on a shared randomness source, while the second
eliminates this assumption using a pseudorandomness generator we prove to exist with a novel generalization of a classical 
result from the study of two-party communication complexity.
We then turn our attention to the easier case where the topology graph is stable, and describe and analyze a new gossip
algorithm that provides a substantial performance improvement for many parameters.
We conclude by studying a relaxed version of gossip in which it is only necessary for nodes to each learn a specified
fraction of the messages in the system.
We prove that our existing algorithms for dynamic network topologies and a single advertising bit
solve this relaxed version up to a polynomial factor faster (in network size)
for many parameters.
These are the first known gossip results for the mobile telephone model, and they significantly expand our understanding of
how to communicate and coordinate in this increasingly relevant setting.
\end{abstract}

%
%
%
%


\newcommand{\polylog}[1]{\text{polylog}(#1)}


\section{Introduction}
\label{sec:introduction}

This paper describes and analyzes new gossip algorithms in
the {mobile telephone model}: an abstraction that
captures the local device-to-device communication capabilities
available in most smartphone operating systems;
e.g., as implemented by services such as Bluetooth LE~\cite{gomez2012overview}, WiFi Direct~\cite{camps2013device}, 
and Apple's Multipeer Connectivity framework~\cite{mark2015peer}.

\paragraph{Motivation.}
Smartphones are a ubiquitous communication platform: there are 
currently over 3.9 billion smartphone subscriptions worldwide~\cite{smartphones}.
Most smartphone communication leverages one-hop radio links to cell towers
or WiFi access points. In recent years, however, the major smartphone operating systems have
included increasingly stable and useful support for local peer-to-peer communication
that allows a device to talk directly to a nearby device (using local radio broadcast) while
avoiding cellular and WiFi infrastructure.

The ability to create these local links,
combined with the ubiquity of smartphones, 
 enables scenarios in which large groups of nearby smartphone users run applications
 that create peer-to-peer
meshes supporting infrastructure-free networking.
There are many possible motivations for these smartphone peer-to-peer networks.
For example, they can support communication in settings where
network infrastructure is {\em censored} (e.g., government protests),
{\em overwhelmed} (e.g., a large festival or march),
or {\em unavailable} (e.g., after a disaster or at a remote event).
In addition, in developing countries,
cellular data minutes are often bought in blocks and carefully conserved---increasing
interest in networking operations that do not require cellular infrastructure. 

To further
validate the potential usefulness of smartphone peer-to-peer networks,
consider the FireChat application, which implements group chat using
smartphone peer-to-peer services. In the few years since its initial release,
it has been widely adopted in over 120 countries and has been used
successfully in multiple government protests, festivals (e.g., at Burning Man, which
is held far from cell towers), and disaster scenarios~\cite{firechat}.

Developing useful applications for this smartphone peer-to-peer setting
requires distributed algorithms that can provide global reliability and efficiency
guarantees on top of an unpredictable collection of local links.
As detailed below, the models that describe this emerging setting are sufficiently
different from existing models that new algorithms and analysis techniques are required.
This paper addresses this need by describing and analyzing new gossip algorithms for this important setting.

\paragraph{The Mobile Telephone Model.}
The mobile telephone model studied in this paper was
introduced in recent work~\cite{ghaffari:2016,newport:2017}.
It is a variant of the classical telephone peer-to-peer model  (e.g.,~\cite{frieze1985shortest,telephone1,telephone2,telephone3,giakkoupis2011tight,chierichetti2010rumour,giakkoupis2012rumor,fountoulakis2010rumor,giakkoupis2014tight})
  modified to better describe the capabilities and
constraints of existing smartphone peer-to-peer services.
The details of the mobile telephone model are inspired, in particular, by the current specifications of Apple's 
Multipeer Connectivity framework~\cite{mark2015peer}: a peer-to-peer service available in every iOS
version since iOS 7 that allows nodes to advertise services, discover nearby advertisers, and attempt to connect
to nearby advertisers, using only local radio broadcast.
(The definition of the classical telephone model,
and differences between the classical telephone and mobile telephone model, are detailed and discussed below
in the related work section.)


In more detail, the mobile telephone model abstracts the basic {\em scan-and-connect}
dynamics of the Multipeer framework as follows.
Time proceeds in synchronous rounds.
In each round, a connected graph describes the underlying network topology for that round.
At the beginning of each round, each device (also called a {\em node} in the following)
learns its neighbors in the topology graph (e.g.,
as the result of a scan). 
Each device can then attempt to initiate a connection with a neighbor.
Each node can support at most one connection---so if multiple nodes attempt
to connect with the same target, only one connection will succeed.
If two nodes connect, they can perform a bounded amount of reliable
communication before the round ends.

We parameterize this model with a {\em tag length} $b\geq 0$.
At the beginning of each round,
each node can choose a {\em tag} consisting of $b$ bits to advertise.
When performing a scan, each node learns both the ids and chosen tags
of its neighbors (where $b=0$ means there are no tags).
These tags can change from round to round.
In our previous study of rumor spreading with parameter $b=1$~\cite{ghaffari:2016}, for example,
at the beginning of a given round, each node that already knows the rumor advertises a $1$ with its tag,
while other nodes advertise a $0$. This simplified the rumor spreading task by enabling nodes that know the rumor
to only attempt to connect to nodes that do not.
This capability of nodes to use tags to deliver limited information to their neighbors
is motivated by the ability of devices to choose and change
their service advertisements in the Multipeer framework.

We also parameterize the model with a {\em stability factor} $\tau \geq 1$.
The underlying network topology must stay stable for at least $\tau$ rounds
between changes. For $\tau=1$, for example,
the network topology can change completely in every round,
while for $\tau=\infty$, the topology never changes.
There exist finer-grained approaches for capturing intermediate levels of stability
(e.g., $T$-interval connectivity~\cite{kuhn2010distributed}), but
in this paper we study only the two extreme cases of fully dynamic and fully stable topologies,
so our simpler stability factor definition is sufficient.
The need to model topology changes is motivated by the inherently
mobile nature of the smartphone setting.


\begin{figure}
\label{fig:results}
\centering
\begin{tabular}{|c|c|c|}
\hline
{\bf Assumptions} & {\bf Algorithm} & {\bf Gossip Round Complexity} \\
\hline 
\hline
\multicolumn{3}{|c|}{Standard Gossip} \\
\hline 
$b=0$, $\tau \geq 1$ & BlindMatch & $O((1/\alpha)k\Delta^2\log^2{n})$ \\
%
%
$b=1$, $\tau \geq 1$ & SharedBit* & $O(kn)$ \\
%
%
$b=1$, $\tau \geq 1$ & SimSharedBit**  & $O(kn + (1/\alpha)\Delta^{1/\tau}\log^6{n})$  \\
%
%
$b=1$, $\tau = \infty$ & CrowdedBin & $O((k/\alpha)\log^6{n})$ \\
\hline \hline
\multicolumn{3}{|c|}{$\epsilon$-Gossip ($0 < \epsilon < 1$)} \\
\hline 
$b=1$, $\tau\geq 1$ & SharedBit* & $O\left(\frac{n\sqrt{\Delta\log{\Delta}}}{(1-\epsilon)\alpha}\right)$ \\
\hline

\end{tabular}
\caption{{\bf A summary of gossip and $\epsilon$-gossip round complexity bounds proved in this paper.}
(In the $\epsilon$-gossip problem, it is assumed that every node starts with a message,
but each node need only learn an $\epsilon$-fraction of the $n$ total messages.)
In the following: $n$ is the network size, $k$ is the number of gossip messages,
$\alpha$ and $\Delta$ are the vertex expansion and maximum degree, respectively,
 of the network topology graph,  $b$ is the tag length, and $\tau$ is the stability factor. All results hold
with high probability in $n$ (i.e., at least $1-1/n$). 
Notice, the result for $b=0$ and $\tau \geq 1$ is the best known result even for the easier case of $b=0$ and $\tau=\infty$.
(*) The SharedBit algorithm (alone among all algorithms studied) requires
shared randomness. (**) The SimSharedBit algorithm is existential in the sense that it depends
on a pseudorandomness generator that we prove exists in Section~\ref{sec:unstable}. }
\end{figure}

\paragraph{Results.}
In this paper,
we describe and analyze new algorithms for the {\em gossip} problem in the mobile
telephone model with respect to different model parameter and algorithm assumptions.
This problem assumes a subset of nodes start with messages (also called {\em tokens}). The goal is to spread these messages to the
entire network.
Gossip is fundamental in distributed computing and is considered particularly important
for ad hoc networks such as the smartphone meshes studied in this paper 
(c.f., the introductory discussion in~\cite{shah:2009}).

Below (and in Figure~\ref{fig:results}) we state and discuss our main results.
In the following, let $n>1$ be the network size and $k, 1\leq k \leq n$, be the number of tokens in the system.
For a given topology graph, 
we use $\alpha$ to describe its vertex expansion (see the model discussion below) and $\Delta$ to describe its
maximum degree.\footnote{If the topology is dynamic, then $\alpha$ is defined as the minimum expansion over all rounds,
and  $\Delta$ is defined as the largest maximum degree over all rounds.}
We assume the topologies are connected.
All round complexity results hold with high probability in $n$ (i.e., probability at least $1-1/n$). 

We start by considering the difficult setting where $b=0$ and $\tau=1$; i.e., nodes cannot use tags and the network topology graph
can change completely in each round.
In Section~\ref{sec:b0}, we describe and analyze a natural strategy for this setting called BlindMatch,
which has nodes select neighbors with uniform randomness to send connection attempts.\footnote{This is essentially
the well-known PUSH-PULL strategy from the classical telephone model with the key exception that in our model
if a node receives multiple connection attempts, only one succeeds. As discussed in the related work and Section~\ref{sec:b0},
this well-motivated model change requires new analysis techniques to understand information propagation.}
We prove that BlindMatch solves gossip in $O((1/\alpha)k\Delta^2\log^2{n})$ rounds.
This bound might seem pessimistic at first glance,
but it is known that disseminating even a single message in the mobile telephone model with this strategy can take $\Omega(\Delta^2/\sqrt{\alpha})$ rounds in some 
networks~\cite{newport:2017}. 
Indeed, this lower bound holds even for the easier assumption that $\tau=\infty$.
Accordingly, we do not consider $b=0$ and $\tau=\infty$ as a distinct case in this paper.
(To provide intuition for why $\Omega(\Delta^2)$ rounds are sometimes necessary, consider two stars centered on $u$ and $v$, respectively,
where each star has around $\Delta$ points and $u$ and $v$ are connected by an edge.
Assume $u$ starts with a gossip message.
For $v$ to receive this message two events must happen: (1) $u$ selects $v$ for a connection; and (2) $v$
accepts $u$'s connection from all incoming connections in that round. The first event occurs with probability $\approx 1/\Delta$,
and because $v$ can expect a constant fraction of its neighbors to send it connection attempts in any given round,
the second event also occurs with probability $\approx 1/\Delta$.)
Our BlindMatch result provides the benchmark against which we attempt to improve with the algorithms that follow.

In Section~\ref{sec:unstable}, we consider the case where $b=1$ and $\tau\geq 1$; i.e., the network can still change completely in each round,
but now nodes can advertise a single bit to their neighbors.
We begin by describing and analyzing an algorithm called SharedBit.
This algorithm assumes a shared randomness source which is used to implement (essentially) a random hash
function that allows nodes to hash their current set of known messages to a single bit to be used as their one-bit advertising tag.
The key guarantee of this function is that nodes with the same sets advertise the same bit,
and nodes with different sets have a constant probability of advertising different bits.
This helps nodes seek out productive connections with neighbors (e.g., connections in which at least one node learns something new).
We prove that SharedBit solves gossip in $O(kn)$ rounds.

We next seek to eliminate the shared randomness assumption.
To do so, we describe SimSharedBit which solves gossip in $O(kn + (1/\alpha)\Delta^{1/\tau}\log^6{n})$ rounds,
without assuming a shared randomness source.
Notice, because $\alpha \geq 2/n$ and $\Delta \leq n$, this solution is always within log factors of the SharedBit
for large $k$, and for small $k$ it is still comparable for many values of $\alpha$, $\Delta$, and/or $\tau$.

The SimSharedBit algorithm depends on a novel generalization of {\em Newman's Theorem}~\cite{newmans}---a well-known
result on  public randomness simulation from the study of two-party communication complexity. 
We prove that there exists an appropriate pseudorandom number generator that can provide
sufficient randomness for the SharedBit strategy.
We then elect a leader in $O((1/\alpha)\Delta^{1/\tau}\log^6{n})$ rounds using an algorithm from~\cite{newport:2017},
and use this leader to disseminate a small generator seed.
We note that our generalization of Newman's Theorem is potentially of standalone interest as the techniques
we introduced can be used to study pseudorandomness in many different graph algorithm settings.

In Section~\ref{sec:stable},
we consider the impact of topology changes on gossip time.
In particular, we consider the case where $b=1$ and $\tau=\infty$; i.e., the network topology is stable.
We describe and analyze CrowdedBin, an algorithm that solves gossip in $O((1/\alpha)k\log^6{n})$ rounds.
This algorithm matches or outperforms the $O(kn)$ round complexity of SharedBit for all $\alpha$ values (ignoring 
log factors). For well-connected networks (e.g., constant $\alpha$), it performs almost a factor of $n$ faster.
These results hint that large increases to stability are more valuable to gossip algorithms than large increases
to tag length (for most of our solutions, increasing $b$ beyond $1$ only improves performance by at most logarithmic factors).

The benefit of stable network topologies is that nodes can transmit larger amounts of information about their
current state to their neighbors by using their single bit advertisement tag over multiple rounds.
CrowdedBin leverages this capability to help nodes efficiently converge on an accurate estimate of $k$---which is not known in advance.
This process depends on nodes testing guesses by throwing their tokens into a number of bins corresponding
to the current guess, and then seeking/spreading evidence of crowding (as established by a new balls-in-bins algorithm described
in Section~\ref{sec:stable}).
%
Once all nodes learn an appropriate guess of $k$,
CrowdedBin deploys an efficient parallel rumor spreading strategy to efficiently disseminate the $k$ tokens.

Finally,  we consider the {\em $\epsilon$-gossip} problem,
which is parameterized with a fraction $\epsilon, 0 < \epsilon <1$,
assumes that $k=n$,
and relaxes the gossip problem to require only that every node receives at least $n\epsilon$ of the
$n$ total tokens.
This variation is useful for settings
where it is sufficient for nodes to learn {\em enough} rumors to complete the task at hand; e.g., when an algorithm requires
responses from only a majority quorum of nodes.

In Section~\ref{sec:egossip}, we re-analyze the SharedBit gossip algorithm from Section~\ref{sec:unstable}.
Deploying a novel argument based on finding productive ``coalitions" of nodes,
we show that SharedBit solves $\epsilon$-gossip in $O\left(\frac{n\sqrt{\Delta\log{\Delta}}}{(1-\epsilon)\alpha}\right)$  rounds.
Recall that SharedBit solves regular gossip in $O(n^2)$ rounds under the $k=n$ assumption.
Therefore, when $\epsilon$ is a constant fraction and the network is well-connected ($\alpha$ is large),
SharedBit solves $\epsilon$-gossip up to a (sub-linear) polynomial factor faster than the standard gossip problem.

\paragraph{Related Work.}
The mobile telephone model used in this paper was first  introduced in a study of rumor spreading
by Ghaffari and Newport~\cite{ghaffari:2016}.
We also recently studied leader election in this same model~\cite{newport:2017}.
As noted, the mobile telephone model 
is a variation of the classical telephone model
(first introduced by Frieze and Grimmett~\cite{frieze1985shortest})
adapted to better describe smartphone peer-to-peer networks.
The mobile model differs from the classical model in two ways:
(1) the classical model implicitly fixes $b=0$ and (typically) $\tau = \infty$; and (2) the classical
model allows nodes  to accept an unbounded number of incoming connections.

It is important to emphasize that most of the well-known bounds in the classical model
depend on this assumption of unbounded connections, and
removing this assumption requires new analysis techniques; c.f., the discussion in~\cite{ghaffari:2016}.
We note that work by Daum~et~al.~\cite{kuhn:bounded} (which preceded~\cite{ghaffari:2016,newport:2017})
also pointed out the dependence of existing telephone model bounds on unbounded concurrent connections.

A fundamental problem in peer-to-peer networks is {\em rumor spreading}, in which
a single message must be disseminated from a designated source to all nodes (this is equivalent to gossip with $k=1$).
This problem is well-understood in the classical telephone model,
where spreading times are often expressed with respect to 
spectral properties of the network topology graph such as 
graph conductance (e.g.,~\cite{giakkoupis2011tight})
 and vertex expansion (e.g.,~\cite{chierichetti2010rumour,giakkoupis2012rumor,fountoulakis2010rumor,giakkoupis2014tight}). 
 This existing work established that efficient rumor spreading is possible with respect to both graph properties in the classical model.
%
In~\cite{ghaffari:2016}, we studied this problem in the mobile telephone model.
We proved that efficient rumor spreading with respect to conductance {\em is not} possible in the mobile telephone model,
but efficient spreading with respect to vertex expansion {\em is} possible.
We then proved that for $b=1$ and $\tau\geq 1$, a simple random spreading strategy solves the problem in $O((1/\alpha)\Delta^{1/\tau}\text{polylog}(n))$
rounds---matching the tight $\Theta((1/\alpha)\log^2{n})$ result from the classical telephone model within log factors for $\tau \geq \log{\Delta}$.
In~\cite{newport:2017},
we built on these results to solve leader election in similar asymptotic time.

Though gossip is well-studied in peer-to-peer models (see~\cite{shah:2009} for a good overview),
little is known about how to tackle the problem in the mobile telephone model,
where concurrent connections are now bounded but nodes can leverage advertising tags.\footnote{It might be 
tempting to simply run $k$ parallel instances of the rumor spreading strategy from~\cite{ghaffari:2016}
to gossip $k$ messages, but this approach fails for three reasons:
(1) our model allows only $O(1)$ tokens to be sent per connection per round;
(2) each of the $k$ instances requires its own advertising tag bit, whereas all of our new gossip results focus
on the case where $b\leq1$;
and (3) nodes do not know $k$ in advance.
Accordingly, most results presented in this paper require substantial technical novelty.}
Finally, we note that there are application similarities between gossip in the mobile telephone model
and existing reliable multicast solutions for mobile ad hoc (e.g.,~\cite{gopalsamy2002reliable})
and delay-tolerant (e.g.,~\cite{burleigh2003delay}) networks. These existing solutions,
however, tend to be empirically evaluated and depend on the ability to predict information
about link behavior (e.g., predicted link duration or an advance schedule of when given links will be present).

%

 \section{Model and Problem}
\label{sec:model}
%
%
%

We describe a smartphone peer-to-peer network
using  the {\em mobile telephone model}.
As elaborated in the introduction,
the basic properties of this model---including
its scan-and-connect behavior, dynamic topologies, 
and the nodes' ability to advertise a bounded tag---are inspired in particular
by the behavior of the Apple Multipeer Connectivity framework for smartphone peer-to-peer networking.

In more detail, we
assume executions proceed in synchronous rounds labeled $1,2,...$.
We assume all nodes start in the same round.
We describe a peer-to-peer network topology in each round $r$ as an undirected connected graph 
$G_r=(V,E_r)$ that can change from round to round, 
constrained by the stability factor (see below).
We call the sequence of graphs $G_1, G_2,...$ that describe the evolving topology a dynamic graph.
We assume the definition of the dynamic graph is fixed at the beginning of the execution.

We assume a computational process (also called a {\em node} in the following) is assigned to each vertex in $V$,
and use $n=|V|$ to indicate the network size.
At the beginning of each round $r$, we assume each node $u$ learns its neighbor set $N(u)$ in $G_r$.
Node $u$ can then select at most one node from $N(u)$ and send a connection proposal.
A node that sends a proposal cannot also receive a proposal.
If a node $v$ does not send a proposal, and at least one neighbor sends a proposal
to $v$, then $v$ can {\em accept} an incoming proposal.
There are different ways to model how $v$ selects a proposal to accept.
In this paper, for simplicity, we assume $v$ accepts an incoming proposal selected
with uniform randomness from the incoming proposals. 
If node $v$ accepts a proposal from node $u$,
the two nodes are {\em connected} and can perform a bounded amount of interactive communication to conclude the round.
We leave the specific bound on communication per connection as a problem parameter.


\paragraph{Model Parameters.}
We parameterize the mobile telephone model with two integers, a {\em tag length} $b\geq 0$
and a {\em stability factor} $\tau \geq 1$. 
We allow each node to select a {\em tag} containing $b$ bits to advertise at the beginning
of each round. That is, if node $u$ chooses tag $b_u$ at the beginning of a round,
all neighbors of $u$  learn $b_u$ before making their connection decisions in this round.
A node can change its tag from round to round. 

We also allow for the possibility of the network topology changing between rounds.
We bound the allowable changes with a stability factor $\tau \geq 1$.
For a given $\tau$, 
the dynamic graph describing the changing topology
 must satisfy the property that at least $\tau$ rounds must pass between any changes to the topology.
For $\tau=1$, the graph can change arbitrarily in every round.
We use the convention of stating $\tau=\infty$ to indicate the graph never changes.


\paragraph{Vertex Expansion and Maximum Degree.}
Several of our results express time complexity bounds with
respect to the {\em vertex expansion} $\alpha$ of the dynamic graph describing the network topology.
To define $\alpha$, we first review a standard definition of vertex expansion
for a fixed static unconnected graph $G=(V,E)$.

For a given $S \subseteq V$,  define the {\em boundary} of $S$, indicated $\partial S$, as follows:
 $\partial S = \{ v\in V \setminus S : N(v) \cap S \neq \emptyset\}$: that is, $\partial S$ is the set
 of nodes not in $S$ that are directly connected to $S$ by an edge in $E$.
Next define $\alpha(S) = |\partial S|/|S|$.
As in~\cite{giakkoupis2014tight,ghaffari:2016}, we define the {\em vertex expansion} $\alpha(G)$ of our static graph $G = (V,E)$
 as follows:
 
 \[  \alpha(G) = \min_{S \subset V, 0 < |S| \leq n/2} \alpha(S). \]
 
 \noindent Notice that despite the possibility of $\alpha(S) >1$ for some $S$, we always have $\alpha(G) \leq 1$.
We define the vertex expansion $\alpha$ of a {\em dynamic} graph $G_1, G_2...$,
to be the minimum vertex expansion over all of the dynamic graph's constituent static graphs
(i.e., $\alpha = \min\{\alpha(G_i) : i \geq 1\})$.

Similarly, we define the maximum degree $\Delta$ of a dynamic
graph to be the maximum degree over all of the dynamic graph's constituent static graphs.

\paragraph{The Gossip Problem.}
The gossip problem assumes each node is provided an upper bound\footnote{For the sake of concision,
the results described in the introduction and Figure~\ref{fig:results} make the standard
assumption that $N$ is a polynomial upper bound on $n$, allowing us to replace $N$ with $n$
within logarithmic factors inside asymptotic notation. In the formal theorem statements for these results,
however, we avoid this simplification and leave $N$ in place where used---enabling a slightly finer-grained
understanding of the impact of the looseness of network size estimation on our complexity guarantees.} $N \geq n$ on
the network size and a unique ID (UID) from $[N]$.
The problem assumes some subset of nodes begins with a gossip message to spread (which we also call a {\em token}).
We use $k$ to describe the size of this subset and assume that $k$ is not known to the nodes in advance.
A given node can start the execution with multiple tokens, but no token starts at more than one node.
We treat gossip tokens as comparable black boxes 
that can only be communicated between nodes through connections (e.g., a node cannot transmit a gossip token
to a neighbor by spelling it out bit by bit using its advertising tags).
If a node begins an execution with a token or has received the token through a connection,
we say that the node {\em owns}, {\em knows} or has {\em learned} that token.
We assume that a pair of connected nodes can exchange  at most $O(1)$ tokens
and $O(\polylog{N})$ additional
bits during a one round connection.

\paragraph{Solving the Gossip Problem.}
The gossip problem requires all nodes to learn all $k$ tokens,
Formally,
we say a distributed algorithm
{\em solves the gossip problem in $f(n,k,\alpha, b,\tau)$ rounds},
if with probability at least $1-1/n$, 
all nodes know all $k$ tokens
 by round $f(n,k,\alpha,b,\tau)$ when
executed in a network of size $n$, with $k$ tokens,
vertex expansion $\alpha$,
tag length $b$, and stability factor $\tau$.
We omit parameters when not relevant to the bound.

\paragraph{Probability Preliminaries.}
The analyses that follow leverage the following well-known probability results:

 \begin{theorem}
For $p \in [0, 1]$, we have $(1-p) \leq e^{-p}$ and $(1+p) \geq 2^p$.
 \label{fact:prob}
 \end{theorem}

\begin{theorem}[Chernoff Bound: Lower Bound Form]
\label{thm:chernoff}
Let $Y = \sum_{i=1}^t X_i$ be the sum of $t > 0$ i.i.d.~random indicator variables $X_1$, $X_2$,..., $X_t$,
and let $\mu = E(Y)$.
Fix some fraction $\delta$, $0 < \delta < 1$. It follows:
\[ \Pr(X \leq (1-\delta)\mu) \leq e^{- \frac{\delta^2 \mu}{2}}. \]
\end{theorem}

\begin{theorem}[Chernoff Bound: Upper Bound Form]
\label{thm:chernoff2}
Let $Y = \sum_{i=1}^t X_i$ be the sum of $t > 0$ i.i.d.~random indicator variables $X_1$, $X_2$,..., $X_t$,
and let $\mu = E(Y)$.
Fix some value $\delta>1$. It follows:
\[ \Pr(X \geq (1+\delta)\mu) \leq e^{- \frac{\delta \mu}{3}}. \]
\end{theorem}

\begin{theorem}[Chernoff-Hoeffding Bound]
\label{thm:chernoff2}
Let $X_1$, $X_2$, ..., $X_t$, be $t \geq 1$ i.i.d.~random indicator variables.
Let $\mu = E(X_i)$ and fix some $\delta > 0$.
It follows:

\[  \Pr\left(  \frac{1}{t}\sum_{i=1}^t X_i \geq \mu + \delta \right) \leq e^{- 2 \delta^2 t}. \]

\end{theorem}

\begin{theorem}[Markov's Inequailty]
\label{thm:markov}
Let $X$ be a nonnegative random variable and $a>0$ be a real number. It follows:

\[  \Pr\left(X \geq a\right) \leq \frac{E(X)}{a}. \]
\end{theorem}

\section{Token Transfer Subroutine}
\label{sec:transfer}

An obstacle to solving gossip in the mobile telephone model is deciding which tokens to exchange
between two connected nodes.
In more detail, once two nodes $u$ and $v$ with respective token sets $T_u$ and $T_v$ connect, 
even if they {\em know} $T_u \neq T_v$, 
they must still identify at least one token $t\notin T_u \cap T_v$ to transfer for this round of gossip
to be useful. Complicating this task is the model restriction that $u$ and $v$
can only exchange $O(\polylog{N})$ bits before deciding which tokens (if any)
to transfer. This is not (nearly) enough bits to encode a full token set
(a simple counting argument establishes that every coding scheme will require $\Omega(N)$ bits
for some sets). Therefore, a more efficient routine is needed to implement this useful token
transfer.

Here we describe a {\em transfer subroutine} that solves this problem and is used by multiple
gossip algorithms described in this paper. 
This routine, which we call $Transfer(\epsilon)$, for an error bound $\epsilon$, $0 < \epsilon < 1$,
is a straightforward application of an existing algorithmic tool
from the literature on two-party communication complexity.
It guarantees the following: if $Transfer(\epsilon)$ is called by two connected nodes $u$ and $v$,
with respective token sets $T_u$ and $T_v$,
and $T_u \neq T_v$, then with probability at least $1-\epsilon$ the smallest token $t$ (by a predetermined token ordering)
that is {\em not} in $T_u \cap T_v$, will be transferred by the node that knows $t$ to the node that does not.
This routine requires $u$ and $v$ to exchange only $O(\log^2{N} \cdot \log{(\frac{\log{N}}{\epsilon})})$ controls bits in addition
 to token $t$. It also assumes some fixed ordering on tokens.

\paragraph{Equality Testing.}
We use one of the many known existing solutions 
to the {\em set equality} (EQ) problem from the study
of two-party communication complexity.
In our setting with $u$ and $v$ (described) above,
these existing solutions provide $u$ and $v$ a way to test the equality of $T_u$ and $T_v$,
and they offer the following guarantee:
if $T_u = T_v$, then $u$ and $v$ will correctly determine their sets are equal with probability $1$,
else if $T_u \neq T_v$ then $u$ and $v$ will {\em erroneously} determine their sets are equal
with probability no more than $1/2$.
These existing solutions assume only private randomness and require $u$ and $v$
to exchange no more than $O(\log{N})$ bits.
A nice property of most such solutions is that each trial is independent.
Therefore, if $u$ and $v$ repeat this test $c$ times, for some integer $c\geq 1$,
then the error probability drops exponentially fast with $c$ to $2^{-c}$.
Let us fix one such equality testing routine and call it $EQTest(c)$,
where parameter $c\geq 1$ determines how many trials to execute in testing the equality.

\paragraph{The Transfer Subroutine.}
We now deploy $EQTest(\epsilon')$,
for $\epsilon' = \lceil \log{(\frac{\log{N}}{\epsilon})}\rceil $, as a subroutine to implement 
the $Transfer(\epsilon)$ routine.
In particular, recall that for a given $u$ and $v$,
we can understand $T_u$ and $T_v$ to both be subsets of
the values in $[N]$ (as each node in the network can label each token with its UID from $[N]$ at the beginning
of the execution).
Our goal is to identify the smallest location value in $[N]$ 
that is in $T_u \cup T_v$ but not in $T_u \cap T_v$.
To do so, we can implement a binary search over the interval $[N]$,
using $EQTest(\epsilon')$ to test the equality of the interval in question between $u$ and $v$.
In more detail:

\bigskip

\noindent {\bf Transfer}$(\epsilon)$:

$a \gets 1$; $b\gets N$

{\bf while} $a\neq b$

$\>\>\>$ $result \gets$ {\bf EQTest}$(\epsilon')$ executed
 on  $T_u\cap [a,\lfloor b/2 \rfloor]$ and $T_v \cap [a,\lfloor b/2 \rfloor]$
 
 $\>\>\>$ {\bf if} $result = notequal$ {\bf then} $b \gets \lfloor b/2 \rfloor$ {\bf else} $a \gets \lfloor b/2 \rfloor +1$

{\em transfer} token $a$ to the other node if you know token $a$

\bigskip

\noindent The above logic implements a basic binary search over the interval $[N]$ to identify the smallest
value in this interval that is in exactly one of the two sets $T_u$ and $T_v$.
If every call to $EQTest$ succeeds then the search succeeds and $Transfer$ behaves correctly.
There are at most $\log{N}$ calls to $EQTest$, each of which fails with probability
$2^{-\epsilon'} \leq \epsilon/\log{N}$.
Therefore, by a union bound, the probability that at least one of the $\log{N}$ calls to $EQTest$ fails
is less than $\epsilon$, as claimed.
From a communication complexity perspective,
each call to $EQTest(\epsilon')$ requires $O(\log{N}\cdot \epsilon') = O(\log{N}\cdot \log{(\log{N}/\epsilon)})$ bits,
and we make $\log{N}$ such calls.
Therefore, the total communication complexity is in $O(\log^2{N}\cdot \log{(\frac{\log{N}}{\epsilon})})$, as claimed.


\section{Gossip with $b=0$ and $\tau \geq 1$}
\label{sec:b0}	

 Here we consider the most difficult case for gossip in our model:
 nodes cannot advertise any information to their neighbors ($b=0$),
 and the network topology graph can change arbitrarily in every round ($\tau = 1$).
 We will study the straightforward strategy in which nodes randomly select neighbors
 for attempted connections and then use the token transfer routine 
 to select tokens to exchange during successful connections.
 We will show this strategy solves gossip
 in $O((1/\alpha)k\Delta^2\log^2{N})$ rounds
 when executed with $k$ tokens in a network graph with expansion $\alpha$ and maximum degree $\Delta$.
 This result might seem pessimistically large at first glance,
 but as shown in~\cite{newport:2017}, there are networks in which simple blind connection strategies like those
 implemented here
 do require $\Omega(\Delta^2/\sqrt{\alpha})$ rounds to spread even a single message.

\paragraph{The {BlindMatch} Gossip Algorithm.}
At the beginning of each round $r\geq 1$,
each node $u\in V$ flips a fair coin to decide whether to be a {\em sender} or a {\em receiver} in $r$.
If $u$ decides to be a sender, it selects a neighbor uniformly from among its neighbors in this round
and sends it a connection proposal. If $u$ decides to be a receiver it waits to receive proposals.
If two nodes $u$ and $v$ connect, they execute the token transfer subroutine 
which attempts to transfer the smallest token in $(T_u(r) \cup T_v(r)) \setminus (T_u(r) \cap T_v(r))$,
assuming such a token exists.

\paragraph{Analysis.}
We now prove the below theorem concerning about the performance of the BlindMatch algorithm.
The proof adapts our recent analysis of leader election strategies in the mobile
telephone model under the assumption that $b=0$~\cite{newport:2017}.
The main contribution of this section, therefore, is less technical than it is the establishment
of a baseline against which to compare the other results studied in this paper.

\begin{theorem}
\label{thm:blindmatch}
The BlindMatch gossip algorithm solves the gossip problem in $O((1/\alpha)k\Delta^2\log^2{N})$ rounds
when executed with tag length $b=0$ in a network with stability $\tau\geq 1$.
\end{theorem}
\begin{proof}
In~\cite{newport:2017}, we
study a leader election algorithm called BlindGossip that essentially matches the behavior of BlindMatch.
As in BlindMatch, this algorithm has each node in each round flip a coin to decide whether or not to send or receive,
and senders choose a neighbor uniformly to send a connection proposal.
If two nodes connect, they transfer the smallest UIDs they have seen so far in the execution.
In~\cite{newport:2017}, 
we prove that this strategy will disseminate the smallest UID in the network
to all nodes in the network in $O((1/\alpha)\Delta^2\log^2{N})$ rounds, with high probability in $N$.
This existing analysis follows the progress of the smallest token in the network showing that after this many
rounds it will have spread to all nodes.

In BlindMatch, by contrast, a connected pair executes the transfer routine to attempt to transfer
the smallest token known by one but not both of the connected nodes.
It follows, therefore, that under the assumption that the transfer routine works correctly
every time it is called, 
BlindMatch will spread the smallest token in the network to all nodes in the time stated above.
Once this has been accomplished, however, we can turn our attention to the second smallest
token (once all nodes know the smallest token, the transfer routine will always transfer the second
smallest when a node that knows the second smallest is connected to a node that does not).
After the above number of rounds, the second smallest token will also have spread.
We repeat this process for all $k$ tokens to get the final  $O((1/\alpha)k\Delta^2\log^2{N})$ time 
claimed above.
\end{proof}


\section{Gossip with $b=1$ and $\tau \geq 1$}
\label{sec:unstable}	

Here we describe and analyze two gossip algorithm that now assume $b=1$.
The first, called SharedBit, assumes shared randomness, while the second, SimSharedBit, does not.
Both solutions offer a substantial time complexity  improvement over the BlindMatch algorithm
for many graph parameters.

\paragraph{Discussion: Shared Randomness.}
For the sake of clarity, we begin by making a strong assumption that we will subsequently eliminate:  
the nodes have access to a shared randomness source.
In more detail, we assume at the beginning of the execution a bit string $\hat r$ of length
$T={O}(N^3\log{N})$
is selected with uniform randomness from the space ${\cal R}$ of all bit strings of this length.
All nodes can access $\hat r$.
This shared random string simplifies the description and analysis
of an efficient gossip algorithm for the assumptions tackled in this section.
 In particular, the key challenge for gossip in this setting is
 generating useful $1$-bit advertising tags in each round. 
 We would like nodes with the same token set to generate
the {\em same} bit (so they will know not to attempt to connect to each other), while pairs of nearby nodes with different token
sets to have a reasonable probability of generating {\em different} bits (so they will know a connection would prove useful).
Shared randomness enables this property as each node can associate the same fresh random bit for each token in a given round,
and the bit advertised for a given set can simply consist of the sum of the bits associated with tokens in the set (mod 2).

\paragraph{Discussion: Eliminating the Shared Randomness Assumption.}
The assumption of shared randomness might be unrealistic in some settings.
With this in mind, we will then proceed to show how to eliminate this assumption by simulating
public randomness using a much smaller number of private random bits that disseminate quickly throughout the network. 
The core strategy of this simulation borrows and expands key ideas from
the proof of {\em Newman's Theorem} (e.g.,~\cite{newmans})---a well-known result
on public randomness simulation from the study of two-party communication complexity.
Our result is existential in the sense that it establishes that there exists 
an efficient simulation of our shared randomness that works well enough.
An equivalent formulation of this result in the language of pseudorandomness is
that there exists a pseudorandom number generator that can generate the needed
number of bits with a seed sufficiently small to fit in our message size bound.

\subsection{Shared Randomness}
\label{sec:unstable:alg}

Here we describe and analyze the SharedBit gossip algorithm.

%
%
\paragraph{The SharedBit Gossip Algorithm.}
Let ${\hat r}$ be a shared random string of length $cN^3(\lceil \log{N} \rceil + 1)$ bits.
We assume nodes partition $\hat r$ into $cN^2$ {\em groups} each consisting
of $N$ {\em bundles} (one for each id that might show up in the network) that each contain $\lceil \log{N} \rceil + 1$ bits.
We label these groups $1,2,...,cN^2$, and label the bundles within a given group $1,2,...,N$.

At the beginning of each round $r \leq cN^2$,
node $u$ must decide which bit to advertise to its neighbors (i.e., what value to select for $b_u(r)$).
If $T_u(r)$ is empty, then $u$ advertises $0$ (i.e., $b_u(r)=0$).
Otherwise, node $u$ calculates its advertisement
by first extracting a shared bit from $\hat r$ to assign to each $t\in T_u(r)$.
In particular, for each such $t\in T_u(r)$, 
$u$ sets its bit, indicated $t.bit$, 
to be the first bit in bundle $t$ of group $r$ from $\hat r$.
Node $u$ then calculates the bit $b_u(r)$ to advertise in this round as follows:

$$b_u(r) = \left(\sum_{t\in T_u(r)} t.bit\right)\mod{2}.$$

If $b_u(r) = 0$ then $u$ will receive connection proposals in this round.
If $b_u(r)=1$ and $u$ has at least one neighbor advertising $0$,
then $u$ will choose one these neighbors with uniform randomness and send it a connection proposal.
To make this random choice, $u$ uses the random bits in positions $2$ to $\lceil \log{N} + 1 \rceil$
in the the bundle corresponding to its id in group $r$ of $\hat r$.\footnote{The reason we have $u$
use shared random bits to select the receiver of its proposal is because it will simplify our subsequent
effort to eliminate shared randomness for this algorithm. There are many straightforward ways a node can use (up to) $\log{N}$
bits to uniformly select a value from a set containing no more than $N$ values.}

If two nodes $u$ and $v$ connect in round $r$,
they will deploy the token transfer subroutine,
with parameter $\epsilon = n^{-c_t}$, for some sufficiently large constant $c_t \geq 1$ we fix in the analysis.
This routine will
identify and transfer the smallest token in $(T_u(r) \cup T_v(r)) \setminus (T_u(r) \cap T_v(r))$, without sending more than $\polylog{N}$ bits
in the interaction (the bound enforced by our model).
Recall, this {\em transfer subroutine} is probabilistic and succeeds in identifying 
a token to transfer with probability at least $1-\epsilon$.
Once the algorithm proceeds past round $cN^2$ it can terminate or fall back
to a simpler behavior (such as our algorithm for $b=0$), or recycle back to the beginning of the shared string.


\paragraph{Analysis.}
Our goal is to  prove the following theorem regarding the SharedBit gossip algorithm:

\begin{theorem}
The SharedBit gossip algorithm solves the gossip problem in $O(kn)$ rounds
when executed with shared randomness and tag length $b=1$, in a network with stability $\tau \geq 1$.
\label{thm:unstable}
\end{theorem}

To setup our analysis, recall that we define $T_u(r)$ for node $u$ and round $r\geq 1$,
to be the set of tokens $u$ knows at the beginning of round $r$,
and use $b_u(r)$ to indicate the bit advertised by $u$ in round $r$.
Also recall that $cN^2$ is the maximum number of rounds for which the shared string $\hat r$ contains bits
(our below analysis will specify the needed lower bound on constant $c\geq 1$ ),
and
that $t.bit$, for a given token $t$ and a fixed round,
 describes the shared random bit extracted from $\hat r$ and assigned to $t$ in this round.

We begin with
 the following lemma,
  which bounds the probabilistic behavior of the advertising tags generated using a given shared $\hat r$.
 
\begin{lemma}
\label{lem:unstable:adv}
Fix two nodes $u,v\in V$, $u\neq v$,
and a round $r$, $1 \leq r \leq cN^2$.
Fix a $r-1$ round execution of SharedBit, and
let $p=\Pr(b_u(r) \neq b_v(r))$ be the probability (defined
over the random selection of the relevant bits in $\hat r$) that $u$ and $v$ generate
 different advertising bits in round $r$.
If $T_u(r) = T_v(r)$ then $p=0$, else if $T_u(r) \neq T_v(r)$, then $p = 1/2$.
\end{lemma}
\begin{proof}
If $T_u(r) = T_v(r)$ then by definition of the algorithm $b_u(r) = b_v(r)$.
We turn our attention, therefore, to the remaining case where $T_u(r) \neq T_v(r)$.
In the following, for a given non-empty token set $T$,
 define:
 
  $$adv_r(T) =\left(\sum_{t\in T} t.bit\right)\mod{2}.$$

And for the case of an empty set, we define by default $adv_r(\emptyset) = 0$.
Fix $T'_u(r) = T_u(r) \setminus T_v(r)$ and $T'_v(r) = T_v(r) \setminus T_u(r)$.
Let $T'_{u,v}(r) = T_u(r) \cap T_v(r)$.
It follows:

\begin{eqnarray*}
 b_u(r) &=& adv_r(T'_u(r)) + adv_r(T'_{u,v}(r)) \mod{2}\\
 b_v(r) &=& adv_r(T'_v(r)) + adv_r(T'_{u,v}(r)) \mod{2}\\
 \end{eqnarray*}
 
Given the above observation,
we note that $b_u(r) = b_v(r)$ if and only if $adv_r(T'_u(r)) = adv_r(T'_v(r))$.
By definition, $T'_u(r)$ and $T'_v(r)$ have no values in common and
at least one of these sets is non-empty. The bits used in these sums are all therefore pairwise
independent and generated uniformly.
The probability that both these sums are equal is exactly $1/2$,
and therefore so is the complementary probability of inequality.
\end{proof}

We next define the following useful potential function that captures
the amount of information spreading still required in the network to solve gossip after a given round:

\[ \forall r\geq 1: \phi(r) = \sum_{u\in V} \left(   k-|T_u(r)|   \right). \]

 Notice that this function is non-increasing (as nodes never unlearn a token),
and once the function evaluates to $0$, there is no more information to spread and therefore gossip is solved.
We now leverage the definition of potential function $\phi$ from above to define what it means for a round to be {\em good}
with respect to making progress with the gossip problem:

\begin{definition}
We say a given round $r \geq1$ is {\em good} if and only if one of the following two properties is true:
(1) $\phi(r) = 0$; or (2) $\phi(r+1) < \phi(r)$.
\end{definition}

The following result leverages Lemma~\ref{lem:unstable:adv} to formalize the key property 
that each round of our algorithm has a reasonable probability of being good
by our above definition.

\begin{lemma}
\label{lem:unstable:good}
For every
round $r$, $1\leq r \leq cN^2$,
the probability that round $r$ is good is at least $1/4$.
\end{lemma}
\begin{proof}
There are two cases depending on the value of $\phi(r)$.
If $\phi(r) = 0$, then by definition this round is good.
Else if $\phi(r) > 0$, we must consider the probability that at least one node learns a new token in this round.
To do so, fix some token $t$ that is not known by all $n$ nodes at the beginning of $r$ (such a token must exist by the assumption that $\phi(r) >0$).
Let $S$ be the nodes that know $t$. 
Because we assume the network topology is connected in each round,
there must be an edge during round $r$ between a node $u\in S$ and a node
$v\in V\setminus S$.

Because $t\in T_u(r)$ and $t\notin T_v(r)$,
we know $T_u(r) \neq T_v(r)$.
By Lemma~\ref{lem:unstable:adv},
the probability that $b_u(r) \neq b_v(r)$ is $1/2$.
Assume this event occurs.
Also assume $b_u(r)=1$ and $b_v(r)=0$ (the opposite case is symmetric).
By the definition of the algorithm, $u$ will attempt to send a proposal in this
round and it has at least one neighbor to choose from to receive this proposal.
Let $v'$ be the neighbor $u$ chooses.
Whether or not $v'=v$,
we know that $v'$ advertised $0$ in this round.
By Lemma~\ref{lem:unstable:adv}, it follows that $v'$ has a different token
set than $u$ in this round.
Indeed, this must be true of $v'$ and {\em any} node that sends it a proposal in this round.

Now that we have established that $v'$ receives at least one proposal,
we know $v'$ will form a connection this round.
As we just noted, this connection will be with a node $u'$ such that
$T_{u'}(r)  \neq T_{v'}(r)$.
Therefore, with high probability in $n$, the transfer subroutine
will successfully identify a missing token to transfer between $u'$ and $v'$---reducing $\phi$.

We have just shown that for $r$ to be good in the case where $\phi(r) > 0$,
it is sufficient that the following two events occur: (1) $b_u(r) \neq b_v(r)$;
and (2) the transfer subroutine between $u'$ and $v'$ succeeds.
The first occurs with probability $1/2$, and the second with high probability,
which is at least $1/2$ for $n>1$ (which must be true if $\phi(r) >0$).
Both events occur, therefore, with probability at least $1/4$---as required.
\end{proof}

We can now leverage Lemma~\ref{lem:unstable:good} to prove Theorem~\ref{thm:unstable}.
The key argument in the following is that $\phi(1) \leq kn$,
therefore $kn$ good rounds are sufficient to solve the gossip problem.
With high probability, $T=\Theta(kn)$ total rounds is sufficient to achieve this goal---{\em assuming}
that $\hat r$ is long enough to supply random bits for $T$ rounds.
To assure this holds we fix the constant $c$ in the definition of $\hat r$
to be at least the constant identified in the analysis below for the definition of $T$ (which turns out to be $32$).

Formalizing this intuition, however, requires some care in dealing with potential dependencies
between different rounds with respect to their goodness.

\begin{proof}[Proof (of Theorem~\ref{thm:unstable})]
The potential function $\phi$ measures the number of missing values over the $n$ total nodes.
Each node can miss at most $k$ values. Therefore: $\phi(1) \leq kn$.
Because $\phi$ is non-increasing, it is sufficient to ask how many rounds are required
to ensure $kn$ good rounds with high probability.
Here we show that $32kn$ rounds are more than sufficient. 
If we fix the constant $c$ used in the definition of $\hat r$ to $32$,
therefore, it follows that $\hat r$ is sufficiently long to supply random bits
for all $32kn$ rounds needed for high probability termination.

Continuing with the proof, let $X_r$, for each round $r\geq 1$, be the random indicator variable that evaluates
to $1$ if and only if round $r$ is good.
Let $Y_t$, for some round count $t \geq 1$, be defined as:

\[ Y_t = \sum_{r=1}^t X_r. \]

The $Y_t$ variable, in other words, measures the number of good rounds in the first $t$ rounds.
By Lemma~\ref{lem:unstable:good},
we know $E(Y_t) \geq t/4$.
Therefore, in expectation, $4kn$ rounds are sufficient to achieve $kn$ good rounds.
To achieve high probability, however, we cannot simply concentrate on this expectation
as there may be dependencies between different $X$ variables (e.g., the outcome in one round might
increase the probability that the next is good).

Because  Lemma~\ref{lem:unstable:good} establishes a lower bound on this probability that holds
regardless of the execution history, we can deploy a  stochastic dominance argument to achieve
our needed result.
In more detail, let $\hat X_r$, for each $r \geq 0$, be the trivial random indicator variable that
evaluates to $1$ with independent probability $1/4$.
Let $\hat Y_t = \sum_{r=1}^t \hat X_r$.
Clearly, $E(\hat Y_t) = t/4$.
Because the $\hat X$ variables are pairwise independent,
we can concentrate on this expectation.
For example, fix $t= 32kn$.
Applying the Chernoff bound from Section~\ref{sec:model} (Theroem~\ref{thm:chernoff})
with $\delta = 1/2$ and $\mu = E(\hat Y_t) = t/4 = 8kn$,
it follows:

\[ \Pr(\hat Y_t \leq 4kn) \leq e^{- \frac{8kn}{8}} \leq e^{-n} < 1/n. \]

That is, for this particular value of $t\in \Theta(kn)$,
the probability that $\hat Y_t$ is less than $kn$ is small in $n$.
We now note that for each $r\geq 1$,
$X_r$ stochastically dominates $\hat X_r$.
It follows that our above bound on $\hat Y_t$ holds for $Y_t$ as well---which is sufficient to conclude the proof.
%
%
%
%
\end{proof}

\subsection{Eliminating the Shared Randomness Assumption}
\label{sec:unstable:derandomize}

Here we discuss how to remove the assumption of shared randomness.
In more detail, we describe {SimSharedBit}, a variation of {SharedBit}
that does not use shared randomness.
We emphasize that this new algorithm is existential instead of constructive.
Formally, it depends on a small set of bit strings, called ${\cal R'}$, that we prove exists but do not explicitly construct.
Accordingly,
our main  theorem statement below references the {\em existence} of a string set ${\cal R'}$ for which SimSharedBit is 
an efficient solution.

The {SimSharedBit} algorithm adds an additive cost of $\tilde{O}(\Delta^{1/\tau}/\alpha)$ rounds to the
existing time complexity of {SharedBit}. For most combinations of $\Delta$, $\tau$, and $\alpha$,
and $k$, 
this additive cost is swamped by the $O(kn)$ time complexity of {SharedBit}.
For the worst-case values of these parameters,
this extra cost can make SimSharedBit up to a factor of $n$ slower than SharedBit 
(e.g., when $k=1$, $\alpha = 1/n$, $\Delta = n-1$, and $\tau=1$).

\paragraph{Strategy Summary.}
The high-level strategy for {SimSharedBit} is to first elect a leader that disseminates
a {\em seed} string that can be used to generate sufficient randomness to run SharedBit.
Notice, the number of shared bits required by SharedBit is much too large to be efficiently
disseminated (our model restricts connections to deliver $\polylog{N}$ bits per round,
while SharedBit requires $\Omega(N^3)$ shared bits).
The {seed} selected and disseminated by the leader, by contrast, is small enough to be
fully transmitted over a connection in a single round.
To prove that there exists a randomness generator
that can extract sufficient randomness for our purpose from seeds of this small size,
we adapt the technical details of Newman's Theorem ( e.g.,~\cite{newmans})
from the simpler world of two-party communication to the more complicated world of $n$ parties
on a distributed and changing network topology.
In more detail, 
we prove the existence of a multiset ${\cal R'}$, containing only poly($N$) bit
strings of the length required for SharedBit, that is {\em sufficiently} random to guarantee
that if a leader chooses $\hat r$ uniformly
from ${\cal R'}$, the {SharedBit} algorithm using shared randomness $\hat r$ is still likely to solve gossip efficiently.
Because ${\cal R'}$ contains only poly($N$) strings,
the leader can identify the string it selected using only polylog($N$) bits (this selection is the {\em seed} it disseminates)---enabling efficient dissemination of this information.
The existential nature of {SimSharedBit} is entirely encapsulated in the existence of this set ${\cal R'}$.

\bigskip

\noindent Below we begin by describing the guarantees of the leader election primitive we will leverage in the {SimSharedBit} algorithm.
We then describe the operation of {SimSharedBit} before proceeding with its analysis.

\paragraph{Leader Election.}
To elect a leader we can deploy the {BitConvergence} leader algorithm
described in our recent study of leader election in the mobile telephone model~\cite{newport:2017}.
When run in a network with expansion $\alpha$, stability factor $\tau \geq 1$, and maximum degree $\Delta$,
this algorithm guarantees with high probability in $N$
to solve leader election in $O((1/\alpha)\Delta^{1/\tau}\polylog{N})$ rounds.
We emphasize that the algorithm does not require advance knowledge of $\alpha$, $\Delta$, or $\tau$---its time complexity
adapts to the network in which it is executed.

To provide slightly more detail about this algorithm, 
in each round, each node identifies a single identifier to be its candidate leader for that round.
To ``solve leader election" means that eventually all candidate leaders in the network have permanently stabilized to the same identifier.
As noted in~\cite{newport:2017},
a trivial extension to the algorithm allows each node to also generate a {\em payload} consisting of polylog($N$) bits
that follows its identifier. Each node now maintains a variable for its current candidate leader and a variable for that candidate's payload.
We will leverage this payload in {SimSharedBit} to carry a pointer to a $\hat r$ value from ${\cal R'}$.
Finally, we note that {BitConvergence} also maintains the useful property that 
the eventual leader will be the node with the {\em smallest} identifier of all participating nodes.
This simplifies our analysis.

\paragraph{The {SimSharedBit} Gossip Algorithm.}
We are now ready to describe the {SimSharedBit} gossip algorithm.
This new  gossip algorithm interleaves the BitConvergence leader election algorithm described above with the logic
from {SharedBit} gossip.
In more detail,
we will prove below the existence of a multiset ${\cal R'}$, containing poly($N$) bit strings,
that is ``sufficiently random" (a concept we will formalize soon)
that it is sufficient for the nodes in the network to agree on a shared string $\hat r$ sampled from ${\cal R'}$,
instead of from the space of all possible strings of the needed length.

In more detail, at the beginning of the execution,
each node selects its own string from ${\cal R'}$ with uniform randomness.
Assume we have fixed in advance a deterministic unique labeling of the poly($N$) strings
in ${\cal R'}$ with the values $1,2,...,|{\cal R'}|$.
Each node can therefore refer to the string it selected with its label.
Following the standard conventions of pseudoranomness,
we call this label the {\em seed} for the string.
Notice, each seed can be described with only polylog($N$) bits.
We take advantage of this small
size by having each node run the leader election algorithm summarized above with this string stored in its payload.
Therefore, once we elect a leader, all nodes also know its seed.

To interleave gossip and leader
election we will treat even and odd rounds differently.
In even rounds, nodes execute the {BitConvergence} leader election algorithm described above,
using their seed as their payload.
In odd rounds, nodes execute the {SharedBit} gossip algorithm.
In each odd round, 
each node uses as the shared string $\hat r$ whatever string from ${\cal R'}$
is pointed to by the seed in their current candidate leader's payload.
In defining ${\cal R'}$ below,
we will fix the length of strings in this set to be slightly longer than the strings used
by SharedBit, so as to capture the extra rounds required for the network to converge
on a single string (the rounds before this point are potentially wasted with respect
to making gossip progress).


\paragraph{Proving the Existence of a Sufficiently Random ${\cal R'}$.}
To prove {SimSharedBit} solves gossip efficiently with high probability,
we must prove that a shared string sampled uniformly from ${\cal R'}$
is sufficiently random that the {SharedBit} logic executed in odd rounds will
still solve gossip with high probability.

To do so, we begin by establishing some preliminary assumptions and definitions.
First, we note that the string $\hat r$ used by {SharedBit}
consists of $t_{SB} = cN^2$ groups consisting of $N$ bundles that in turn
each contain $t_{b} =(\lceil \log{N} \rceil + 1)$ bits.
The algorithm consumes bits from one group per round, and the analysis
of {SharedBit} requires at most $t_{SB}$ rounds worth of shared randomness to terminate with high probability.

For {SimSharedBit}, we will need to extend this length to account for
the early rounds in the execution when leader election has not yet
converged, and therefore we cannot yet guarantee useful progress for the gossip
logic executing in the odd rounds.
For the worst case values of $\alpha$, $\tau$, and $n$,
{BitConvergence} requires no more than $t_{BC} = O(N^2\polylog{N})$ rounds to converge.
Therefore we  extend the length of shared bit strings to consist of $t_{SSB} = t_{SB} + t_{BC} = {O}(N^2\polylog{N})$ groups.
This ensures that after leader election converges we still have at least the full $t_{SB}$ rounds of randomness needed
for the analysis of {SharedBit} to apply.
At the risk of slightly overloading previous notation, 
we will use ${\cal R} = \{0,1\}^{t_{SSB} \cdot N\cdot t_b}$ to refer to the set of all bit strings of length $t_{SSB} \cdot N\cdot t_b$---the
maximum size shared string needed to give nodes time to converge to a leader and then subsequently solve gossip with the
leader's shared string.
The shared strings used in {SimSharedBit} come from ${\cal R}$.

Next, for a given network size $n>1$,
let $\mathbb{G}(n)$ be the set containing every $t_{SB}$-round dynamic graph defined over $n$ nodes.
That is, if we run our algorithm for $t_{SB}$ rounds in a network of size $n$,
it will be executed in some dynamic graph ${\cal G} \in \mathbb{G}(n)$.
Let ${\cal A}(n)$ be the set containing every assignment of  token sets to the $n$ nodes in a network of size $n$.
We define ``assignment" to capture two key pieces of information: (1) which nodes in the network started with a token;
and (2) which of these tokens does each node know at the moment.
Formally, a given $A\in {\cal A}(n)$ can be described as a function
from $[n]$ to $2^n$.\footnote{This function maps each of the $n$ nodes to some subset of $[1,n]$ indicating
the tokens that node knows. The set of nodes that started with a token according to this assignment
is the set of nodes that have a token show up somewhere in the assignment function's range.}

For each network size $n \in [2,N]$,
round $\ell \in [1,t_{BC}]$,
dynamic graph ${\cal G}\in \mathbb{G}(n)$,
 token assignment $A\in {\cal A}(n)$,
and shared bit string $\hat r \in {\cal R}$:
let $Z(n,\ell,{\cal G}, A, \hat r)$ be the random indicator variable
that evaluates to $0$ if {SharedBit}
solves gossip when run in a network of size $n$,
 starting with token assignment $A$,
and 
executing for $t_{SB}$ rounds in dynamic graph ${\cal G}$,
using the shared random bits from groups $\ell$ to $\ell + t_{SB}$ in $\hat r$.
It otherwise evaluates to $1$.
(In the evaluation of $Z$,
assume that the probabilistic token transfer subroutine used by {SharedBit} 
always works correctly.)
Notice, we are using $0$ to indicate a positive outcome (gossip works),
 and a $1$ to indicate a negative outcome (gossip failed).

In other words, $Z(n,\ell,{\cal G}, A, \hat r)$ answers the following question (with $0$ indicating {\em yes}) :

\begin{quote}
If we assume we are in a network of size $n$,
and that leader election converges to a single leader at round $\ell$,
and this leader points toward shared string $\hat r$,
and that at this point the tokens in the network are spread according to $A$:
will the {SharedBit} logic solve gossip sometime in  the next
$t_{SB}$ rounds, using the corresponding bits from $\hat r$,
assuming the graph evolves as ${\cal G}$ during this round interval? 
\end{quote}

Our analysis of {SharedBit} tell us that if we select $\hat r$
uniformly from ${\cal R}$,
with high probability: $Z(n,\ell,{\cal G}, A, \hat r) = 0$.
Our goal is to prove that there exists a multiset ${\cal R'}$,
made up of values from ${\cal R}$,
such that ${\cal R'}$ only contains poly($N$) strings,
and yet if we select $\hat r$ uniformly from ${\cal R'}$,
the probability $Z(n,\ell,{\cal G}, A, \hat r) = 0$ remains high.
In particular, if $\epsilon$ is an upper bound on the small failure probability of {SharedBit} gossip
when run in a setting with shared randomness,
then we show the probability that $Z$ evaluates to $1$ when drawing $\hat r$ from our multiset ${\cal R'}$
is at most only a constant factor larger.
We formalize this goal with the following lemma.
We emphasize that this setup (analyzing the probability that $Z$ evaluates to $1$ with our reduced ${\cal R'}$)
comes from the proof of Newman's Theorem. We are generalizing this approach, however, to account
for multiple nodes operating on a dynamic graph starting from an arbitrary round within a larger interval,
with an arbitrary distribution of gossip tokens:

\begin{lemma}
\label{lem:z}
There exists a multiset ${\cal R'}$ of size $N^{\Theta(1)}$ containing values from ${\cal R}$,
such that for every $n \in [2,N]$,  $\ell \in [1,t_{BC}]$, ${\cal G}\in \mathbb{G}(n)$ and $A\in {\cal A}(n)$, it follows:
\[\Pr_{\hat r \gets {\cal R'}}\left(  Z(n,\ell,{\cal G}, A, \hat r) = 1  \right) <  2\epsilon,\]

\noindent where $\epsilon=N^{-c}$ (for some constant $c \geq 1$) is an upper bound on the failure probability of {SharedBit} gossip
when executed with shared randomness.
\end{lemma}
\begin{proof}
Fix some network size $n\in [2,N]$, leader election termination round $\ell \in [1,t_{BC}]$, ${\cal G} \in \mathbb{G}(n)$ and $A\in {\cal A}(n)$.
Consider an experiment in which we uniformly select $t$ values $r_1,r_2,...,r_t$ from ${\cal R}$ (with replacement),
where $t > 0$ is a value defined with respect to $N$ that we fix below.
Let $X_i$ be the random indicator 
 variable defined as $X_i = Z(n,\ell,{\cal G}, A, r_i)$.
 That is, $X_i = 0$ if {SharedBit} solves gossip using the relevant bits in $r_i$ in ${\cal G}$
 starting with assignment $A$.
 By Theorem~\ref{thm:unstable} and our definition of $t_{SB}$ (which captures the worst case time complexity from this theorem),
 we know $X_i = 0$ with probability at least $1-\epsilon$.
 Therefore:
 
 \[ E(X_i) = 0\cdot\Pr(X_i=0) + 1\cdot\Pr(X_i=1) \leq \epsilon. \]
 
Note that these random variables $X_1, X_2, ..., X_t$ are i.i.d.~as they are each
 determined by a random string selected with uniform and independent randomness with replacement from a common set.
 It follows that we can apply a Chernoff-Hoeffding bound (Theorem~\ref{thm:chernoff2} from Section~\ref{sec:model})
to $X_1$, $X_2$, ..., $X_t$ to prove that their average value is unlikely to deviate too much from the expected average.
In more detail, let $\mu = E(X_i)$.
This bound tells us that for any $\delta > 0$: 

\[  \Pr\left(  \frac{1}{t}\sum_{i=1}^t X_i \geq \mu + \delta  \right) \leq e^{- 2 \delta^2 t}. \]

Fix $\delta = \epsilon$ and $t = N^{\beta}/\epsilon^2$, for a constant $\beta \geq 1$ we will define below.
We say for our fixed choice of $n$, $\ell$, ${\cal G}$ and $A$,
that a given selection of $t$ strings from ${\cal R}$
is {\em bad} if $\frac{1}{t}\sum_{i=1}^t X_i \geq p = 2\epsilon$.
For our fixed values of $\delta$ and $\epsilon$,
and our above bound,
we know our random choice of strings is bad with probability no more than
$e^{-2 N^{\beta}} < 2^{-N^{\beta}}$.
Put another way, 
for a fixed network size, leader election termination round, dynamic graph and token assignment,
we are very unlikely to have made a bad selection of strings.

Now we consider other values for our parameters.
We know there are no more than $N$ choices for $n$ and $c'N^2\polylog{N}$ choices for $\ell$, for some constant $c' \geq 1$.
For a given $n$,
we can bound $\mathbb{G}(n)$ as

\[ |\mathbb{G}(n)| < (2^{n^2})^{t_{SB}} = 2^{n^2\cdot t_{SB}} \leq 2^{N^{\gamma}},\]

\noindent for some small constant $\gamma \approx 4$. And to bound ${\cal A}(n)$, we note:

\[  |{\cal A}(n)| \leq (2^n)^n \leq 2^{n^2} \leq 2^{N^2}.  \]

\noindent The total number of combinations of $n$, $\ell$, ${\cal G}$ and $A$ values,
therefore,
is upper bounded by: 

\begin{eqnarray*}
N\cdot (c'N^2\polylog{N})  \cdot 2^{N^{\gamma}} \cdot 2^{N^2} 
&\leq&
 c' \cdot 2^{ \log{N^3}+ \log{(\polylog{N})}   + N^{\gamma} + {N^2}   }   \\
 & \leq & 2^{N^{\gamma \cdot c''}}
\end{eqnarray*}

\noindent for some constant $c'' \geq 1$.
Given this upper bound value, we fix the constant $\beta$ used in the definition of $t$
to be some constant strictly greater than $c'' \cdot \gamma$ (say, $\lceil c''\cdot \gamma + 1\rceil$).

We now apply the probabilistic method to prove the existence of a selection of $t$ values from
${\cal R}$ that is {\em not} bad for any of the possible combinations of network sizes,
leader election termination points, graphs and token assignments.
To do, note that the probability of a given selection being bad for a fixed set of parameters
was shown above to be less than $2^{-N^{\beta}}$.
By applying a union bound over the less than $2^{N^{c''\cdot \gamma}}$ combinations of parameters,
the probability that there exists {\em at least} one such combination for which our selection
is bad is less than: $(2^{N^{c''\cdot\gamma}})\cdot(2^{-N^{\beta}}) < 1$.

It follows that there exists at least one collection of $t$ values from ${\cal R}$ that is not bad
for every combination of the relevant parameters.
Let us call this multiset of $t$ values ${\cal R'}$. 

The definition of being not bad for a given graph and assignment is that: 
$\frac{1}{t}\sum_{i=1}^t X_i \leq 2\epsilon$.
It follows that $\sum_{i=1}^t X_i \leq 2t\epsilon$.
From this it follows that at most a $2\epsilon$ fraction of the $X_i$ values evaluate to $1$.
Therefore, if we uniformly sample a string $r_i$ from ${\cal R'}$,
the probability that $X_i = 0$ is at least $1-2\epsilon$, as required by the lemma statement.

To conclude the proof, we must show that $|{\cal R'}| = t$ is in poly($N$).
We earlier fixed: $t=N^{\beta}/\epsilon^2$, where $\beta = \Theta(1)$ and $\epsilon = N^{-c}$ for a constant $c\geq 1$.
It follows that $t=N^{\beta+2c} = N^{\Theta(1)}$.
 \end{proof}

\noindent We now leverage Lemma~\ref{lem:z} to prove our main theorem concerning SimSharedBit:

\begin{theorem}
\label{thm:ssb}
There exists a bit string multiset ${\cal R'}$ of size $N^{\Theta(1)}$,
such that
the SimSharedBit gossip algorithm using this ${\cal R'}$ as its source of simulated shared bit strings
solves the gossip problem in $O(kn + (1/\alpha)\Delta^{1/\tau}\log^6{N})$ rounds when
executed with tag length $b=1$ in a network  with stability $\tau \geq 1$.
\end{theorem}
\begin{proof}
Fix the multiset ${\cal R'}$ proved to exist in Lemma~\ref{lem:z}.
We now study the performance of SimSharedBit using this multiset as the source of shared random
strings selected by leader candidates.

First, we note that by Theorem~\ref{thm:unstable},
we know that {SharedBit} gossip solves gossip in $O(kn)$ rounds with high probability.
In~\cite{newport:2017},
we proved that {BitConvergence} leader election solves leader election in 
$O( (1/\alpha)\Delta^{1/\tau}\log^6{N})$ rounds with high probability.
In Section~\ref{sec:transfer}, we proved that the transfer routine succeeds with high probability.
By a union bound, we can therefore assume that with (slightly less) high probability the transfer routine works every
time it is called in a poly($N$) round execution.

Let $\epsilon$ be the smallest of these three small failure probabilities.
In a given execution of {SimSharedBit},
it follows (by a union bound) that the probability that the transfer routine fails at least
once, or {BitConvergence} fails to elect a leader in the provided time bound,
is less than $2\epsilon$.

Assume neither of these two bad events occur.
We now study the probability that {SimSharedBit},
running with a $\hat r$ selected uniformly by the node with the smallest ID from the
${\cal R'}$, starting from the round right after leader election succeeds,
and runnings on the given dynamic graph for the execution.
By Lemma~\ref{lem:z},
the probability that SimSharedBit fails to solve gossip is also less than $2\epsilon$.

A final union bound on these two failure probabilities establishes that the probability
SimSharedBit gossip fails is less than $4\epsilon$,
and therefore it succeeds with probability at last $1-4\epsilon$.
So long as we set the constant factors in the time complexity of SharedBit, BitConvergence,
and the transfer routine, to ensure that $\epsilon \leq \frac{1}{4N}$, 
SimSharedBit succeeds with high probability.
\end{proof}

%
%
%
%
%
%
%
%
%
%
%


\section{Gossip with $b=1$ and $\tau = \infty$}
\label{sec:stable}

Here we describe and analyze a gossip algorithm that requires
only $\tilde{O}(k/\alpha)$ rounds when executed with $b=1$ and a stable network (where $\tilde{O}$ hides $\polylog{N}$ factors).
Because $\Omega(k)$ is a trivial lower bound for gossip $k$ messages in our model,
this algorithm is optimal for larger $\alpha$. Recall that for $\tau \geq 1$ our best solution required $O(kn)$ rounds.
This algorithm matches this time for the worst-case $\alpha$ values but then improves over it as $\alpha$ increases.
For constant $\alpha$, this algorithm performs a factor of $n$ faster (ignoring log factors).
These results indicate that network stability is valuable from a gossip algorithm perspective.
Notice, for the sake of presentation clarity, the algorithm analysis that follows does not attempt to optimize the polylogarithmic factors
multiplied to the leading $k/\alpha$ term. 

\paragraph{Discussion: Crowded Bins}
We call this algorithm CrowdedBin gossip.
This name comes from a core behavior in the algorithm in which nodes toss
their tokens into a fixed number of bins corresponding to their current estimate $\hat k$ of $k$ (the number of tokens in the network).
Nodes do not know $k$ in advance. Determining this value is crucial to enabling efficient parallel dissemination of their
tokens.
Leveraging a new balls-in-bins analysis, we upper bound the number
of tokens in any given bin {\em if} the estimate $\hat k$ is sufficiently large.
The nodes therefore search for crowded bins as evidence that they need a larger estimate of $k$.
This mechanism provides a way to check that a current guess $\hat k$ is too small while only
paying a time complexity price relative to $\hat k$ (as there are only $\hat k$ bins required
to check for crowding). Because the sequence of guesses we try are geometrically increasing,
the cost of checking estimates smaller than $k$ will sum up to $\tilde{O}(k)$.

\paragraph{Discussion: Spreading Bits versus Spreading Tokens.}
We also emphasize that the CrowdedBin algorithm makes a clear distinction between
propagating information using the advertising bits and propagating the tokens themselves
(which are treated as black boxes, potentially large in size, that require a pairwise connection
for transfer). Combining the stability of the network with each node's ability to advertise a bit
to {\em all} its neighbors in each round,
nodes first attempt to stabilize to a consistent and accurate estimate of $k$,
and a consistent set of {\em tags} describing the network's tokens.
Once stabilized, this information can then support the efficient spreading 
of the tokens, link by link, to the whole network.

\paragraph{The PPUSH Rumor Spreading Strategy.}
The CrowdedBin algorithm uses a simple rumor spreading strategy called PPUSH as a subroutine to help
spread tokens once the network has stabilized.
This algorithm was introduced in our earlier study of rumor spreading in the mobile telephone model~\cite{ghaffari:2016}.
PPUSH assumes a subset of nodes start with a common rumor $m$,
and the goal is to spread $m$ to all nodes.
It requires $b \geq 1$.

In more detail,
the strategy PPUSH works as follows:
(1) at the beginning of each round, if a nodes knows $m$ (i.e., it is {\em informed}), it advertises bit $1$, otherwise if it does not know $m$ (i.e.,
it is {\em uninformed}), it advertises bit $0$;
(2) each informed node that has at least one uninformed neighbor in this round, 
chooses an uninformed neighbor with uniform randomness and attempts to form a connection to spread the rumor.
In~\cite{ghaffari:2016}, we proved the following key result about the performance of PPUSH:

\begin{theorem}[Adapted from~\cite{ghaffari:2016}]
With high probability in $N$: PPUSH succeeds in spreading the rumor to all nodes in $O(\log^4{N}/\alpha)$
rounds when executed in the mobile telephone model with $b\geq 1$, $\tau = \infty$,
and a topology graph with expansion $\alpha$.
\label{thm:ppush}
\end{theorem}

\noindent We will leverage this theorem in our analysis of our gossip algorithm.
We also use the following useful property proved in~\cite{ghaffari:2016} which relates network diameter to expansion:\footnote{The actual
result we proved in~\cite{ghaffari:2016} is that it is always possible to spread a rumor in $O(\log{n}/\alpha)$ rounds in the mobile
telephone model in a graph with expansion $\alpha$. The rumor spreading time in a given network can never be smaller than
the network diameter, which provides a trivial lower bound on the problem.}

\begin{theorem}[Adapted from~\cite{ghaffari:2016}]
\label{thm:diam}
Fix a connected graph with $n$ nodes, expansion $\alpha$, and diameter $D$.
It follows that $D = O(\log{n}/\alpha)$.
\end{theorem}

\subsection{The {CrowdedBin} Gossip Algorithm}
\label{sec:alg:stable}

We divide our description of this analysis into several named parts to clarify its presentation.
In the following, we assume each node $u\in V$ identifies itself with a {\em tag} $t_u$ 
chosen uniformly from the space $\{1,2,...,N^{\beta}\}$,
where $\beta \geq 2$ is constant we fix in our analysis.
Let $\ell = \beta \log{N}$ be the number of bits needed to describe a tag.
To simplify notation, we assume in the following that $N$ is a power of $2$.

\paragraph{Parallelizing Instances.}
Nodes do not know in advance the value of $k$ (the number of tokens in the system).
They consider $\log{N}$ {\em estimates} of $k$: $k_1, k_2,...,k_{\log{N}}$,
where each $k_i = 2^i$.
The nodes run in parallel a separate gossip {\em instance} for each estimate.
We use the notation {\em instance $i$} to refer to the instance corresponding to estimate $k_i$.
In order to run $\log{N}$ instances in parallel, each node uses $\log{N}$ rounds to simulate one round
each of the $\log{N}$ instances.
That is, nodes divide rounds into {\em simulation groups}
consisting of $\log{N}$ rounds.
Round $j$ of simulation group $i$
is used to simulate round $i$ of instance $j$.

\paragraph{Instance Schedules.}
Each instance $i$ groups its rounds into {\em blocks} containing $\ell + \log{N}$ rounds each.
It then groups these blocks into {\em bins} containing $\gamma \log{N}$ blocks each,
where $\gamma > 1$ is a constant we fix in our analysis below.
Finally, it groups the bins into {\em phases} consisting of $k_i$ bins each.
In other words, the schedule for instance $i$ is made up of phases,
where each phase has $k_i$ bins, which are each made up of $\gamma \log{N}$ blocks,
which each contain $\ell+\log{N}$ rounds: adding to a total of $\gamma(\beta+1) k_i \log^2{N}$  total rounds per phase.

\paragraph{Initialization.}
Each node $u\in V$ that begins an execution of the CrowdedBin algorithm with a gossip token,
independently selects a bin for its token for each of the $\log{N}$
instances.
That is, for each instance $i$,
$u$ selects a bin $b_u(i)$ with uniform independent randomness from $\{1,2,...,k_i\}$.
Each node $u$ also maintains, for each instance $i$,
and each bin $j$ for this instance,
a set $T_u(i,j)$ containing the tags it has seen so far for tokens in bin $j$ in instance $i$.
For each instance $i$,
if node $u$ has a token it initializes $T_u(i, b_u(i))= \{t_u\}$ (i.e., it places its own tag in the bin it selected for that instance).
Node $u$ also maintains a set $Q_u$ containing the tokens it has received
so far, where each token in $Q_u$ is also labeled with its tag.
Finally, each node $u$ maintains a variable $est_u$, initialized to $1$,
which describes the current instance node $u$ is participating in.

\paragraph{Participation.}
Each node will only participate in a single instance at a time,
and it will only participate in complete phases of an instance.
In more detail, if some instance $i$ starts a new phase in round $r$,
and some node $u$ has $est_u = i$ at the start of round $r$,
node $u$ is now committed to participate in this full phase of instance $i$.
As we will detail, its estimate cannot change again until this phase completes.

To participate in a phase of instance $i$,
node $u$ does the following.
First, for each bin $j$, $1 \leq j \leq k_i$,
$u$ orders the tags in $T_u(i,j)$ (if any) in increasing order.
It will use the first $\ell$ rounds of the first
block to spell out the smallest such tag,
bit by bit, using its advertising bits (here the assumption that $b\geq 1$ is needed).
It will then use the first $\ell$ rounds of the
second block to spell out the second smallest tag, and so on.
There are $\gamma \log{N}$ total blocks in this bin.
If $u$ knows more than this many tags for this bin, it transmits only the first $\gamma \log{N}$.
Node $u$ transmits all $0$'s during the blocks in this bin for which it has no tags to advertise
(here is where we use the assumption that the smallest possible tag is $1$---preventing a block of
all $0$'s from being mistaken for a tag.)

During the rounds dedicated to bin $j$, node $u$
also collects the bits advertised by its neighbors in each block.
If it learns of a tag $t_v$ that is not currently in $T_u(i,j)$,
it will put it aside and then add it to this set once the rounds dedicated
to bin $j$ in this phase conclude. 

We have only so far described what node $u$ does during the first $\ell$ rounds for each block
in our fixed instance $j$.
During the remaining $\log{N}$ rounds in these blocks,
$u$ will attempt to disseminate 
the actual tokens corresponding to the tags advertised
(here we emphasize the difference between spelling out the bits of
a tag using advertising bits and actually transmitting a token, which requires
two nodes to form a connection).
In more detail, $u$ executes the PPUSH rumor spreading strategy discussed above
during the last $\log{N}$ rounds of each block in the current bin.
In more detail, for a given block $h$ in this bin, if $u$ advertised tag $t$ in the first $\ell$
rounds of this block, {\em and} $u$ actually has the token
corresponding to tag $t$ in $Q_u$,
it executes PPUSH in the remaining rounds of this block using this 
token as the rumor and advertising $1$ (i.e., it runs PPUSH
with the status of an already informed node).
Otherwise, node $u$ runs PPUSH advertising $0$ (i.e., it runs the
PPUSH as an uniformed node).

\paragraph{Increasing Size Estimates.}
A core behavior in this algorithm is how nodes upgrade their current estimate of the value $k$
(stored in $est_u$ for each node $u$).
As described above, each node initializes their estimate to $1$.
As described below, these estimates can only grow during an execution.
We call an increase in this estimate at a given node an {\em upgrade.}
There are two events that trigger an upgrade at a given node $u$.

The first event is that node $u$ sees ``activity" on an instance $i' > est_u$,
where $est_u$ is its current estimate.
The term ``activity" in this context means seeing a $1$-bit advertised
in an instance $i'$ round. 
If this event occurs, then $u$ knows that some other node has already increased
its estimate beyond $est_u$, so $u$ should upgrade its estimate as well.
The second event is that node $u$ fills a bin in its current estimate.
That is, there is some bin $j$ such that $|T(est_u, j)| \geq \gamma \log{N}$.
We call this event a {\em crowded bin},
and $u$ can use this as evidence that $est_u$ does not have enough bins
for the number of tags in the system and therefore $est_u$ is too small of an estimate for $k$.
If this event occurs, $u$ will increase $est_u$ by $1$ (unless $est_u$ is already
at its maximum value in which case it will remain unchanged.).

Recall, as specified above, that if a node $u$ increases its estimate $est_u$ to a new value,
it will complete the phase of whatever instance it was participating in before switching to the new estimate moving forward.
This restriction simplifies the analysis that follows.

%

\subsection{Analysis}

In the following analysis, let $D$ be the diameter of the fixed underlying topology graph.
Some of intermediate results below will reference $D$.
Our final result, however, will be expressed only with respect to $\alpha$ to maintain comparability
to earlier results defined for non-stable networks in which $D$ is not well-defined.

At the beginning of an execution each node randomly assigns a tag from $\{1,2,...,N^{\beta}\}$
to its token, and then randomly assigns the token to a bin in each of the $\log{N}$ instances.
We call the global collection of these assignments for a given execution a {\em configuration.}
Fix a configuration.
We call a given instance $i$ of this configuration, $1\leq i \leq \log{N}$, {\em crowded},
if the configuration has an instance $i$ bin with at least $\gamma \log{N}$ unique tags assigned to it.
The {\em target} instance for our fixed configuration is the smallest instance $i$ that is {\em not} crowded.
If every instance is crowded, then we say the {\em target} instance is {\em undefined}.
We begin our analysis by defining what it means for a configuration to be {\em good} with respect to these terms:

\begin{definition}
A configuration is {\em good} if and only if it satisfies the following two properties:
(1) every token is assigned a unique tag; and (2) the target instance $i$ is defined,
and $k_i \leq 2k$.
\end{definition}

A direct corollary of the above definition is that if a configuration is good,
and $i$ is the target,
then $k_i > k/(\gamma \log{N})$. We now bound the probability that the nodes generate a good configuration.
We will show that increasing the constant $\beta$,
used to define the space $\{1,2,...,N^{\beta}\}$ from which tags are drawn,
and the constant $\gamma$,
used to define the number of blocks per bin, 
increases the high probability that a configuration is good.
To make this argument we begin by proving a non-standard balls-in-bins argument
that will prove useful to our specific algorithm's behavior.

\begin{lemma}
\label{lem:balls}
Fix some constant $\gamma \geq 9$.
Assume $k$ balls, $1 \leq k \leq N$, are thrown into $k' \geq k$ bins with independent and uniform randomness.
The probability that at least one bin has at least $\gamma \log{N}$ balls, 
is less than $1/N^{(\gamma/3) -2}$.
\end{lemma}
\begin{proof}
Label the balls $1,2,...,k$ and the bins $1,2,...,k'$.
Let $b_1$ be bin in which ball $1$ is thrown.
We now calculate the expected number of other balls to land in $b_1$.
To do so, for each ball $i > 1$, 
let $X_1$ be the random indicator variable that evaluates to $1$ if $i$ lands in $b_1$
and otherwise evaluates to $0$.
Let $Y_{b_1} = \sum_{1< i\leq k} X_i$ be the total number of additional balls to
land in $b_1$.
By linearity of expectation and the observation that $E(X_i) =1/k' \leq 1/k$,
it follows that $\mu = E(Y_{b_1}) < 1$.

By definition of the process, $X_i$ and $X_j$ are independent for $i \neq j$.
We can therefore apply an upper bound form of a Chernoff Bound  (Theorem~\ref{thm:chernoff2})
to concentrate near this expectation.
In particular, define $\delta = (\gamma \log{N} - 2)/\mu$.
Notice, $\delta > (\gamma \log{N} -2) > 1$.
We can therefore apply Theorem~\ref{thm:chernoff2}
to $Y=Y_{b_1}$, and our above definitions of $\delta$ and $\mu$.
It follows
that:

\begin{eqnarray*}
 \Pr(Y_{b_1} \geq (1+\delta)\mu)   & \leq & \text{exp}\{- (\gamma \log{N} -2)/3    \}  \\
 & = & \text{exp}\{ -((\gamma/3)\log{N} - 2/3)           \} \\
 & = & \text{exp}\{ -((\gamma/3)\ln{N}\log{e} - 2/3)           \} \\
 & < & \frac{e^{2/3}}{e^{(\gamma/3)\ln{N}}} \\
 & < & 2/N^{\gamma/3} \\
 & \leq &1/N^{\gamma/3 - 1}
\end{eqnarray*}

Notice, $(1+\delta)\mu = \mu + (\gamma \log{N} - 2)$, and $\mu = 1/k' \in (0,1)$. 
Therefore, we can interpret the above bound saying that the probability that $b_1$ has at least $\gamma \log{N} -1$ extra balls is less
than $1/N^{\gamma/3 -1}$.
When we add in ball $1$, which by definition is also in $b_1$,
we get that the probability that $b_1$ has at least $\gamma \log{N}$ balls  is also less than $1/N^{\gamma/3 -1}$.
By symmetry, the same result holds for $b_2$ through $b_k$ as well.
There are dependencies between the outcomes in different bins,
but we can dispatch this issue by applying a union bound
over the $k\leq N$ occupied bins,
which provdes that the probability {\em at least} one bins
has more than $\gamma \log{N}$ balls is less than $N/N^{(\gamma/3) - 1} = 1/N^{(\gamma/3) - 2}$.
\end{proof}

\begin{lemma}
\label{lem:good}
Fix some constant $c\geq 1$.
For a tag space constant $\beta \geq c+3$,
and a bin size constant $\gamma \geq 3c + 9$,
the nodes generate a good configuration with probability at least $1-1/N^c$.
\end{lemma}
\begin{proof}
There are two parts to the definition of {\em good}.
The first requires each tag to be unique.
The probability that there is at least one collision among the tag chocies, given that no more than $N$ tags
are drawn from $N^{\beta}$ options, can be loosely upper bounded as $1/N^{\beta -2}$.
If we define $\beta = c+3$ then this failure probability is less than $1/N^{c+1}$.

The second part of the definition requires that the target instance is defined and it is not too large compared
to the actual number of tokens, $k$.
Let $\hat i = \text{argmin}_{1 \leq i \leq \log{N}}\{k \leq k_i\}$.
That is, $k_{\hat i}$ is the smallest estimate of $k$ considered by our algorithm that is at least as large as $k$.
Because our estimates grow by a factor of $2$, we know that $k_{\hat i} < 2k$.
If we can show that $k_{\hat i}$ is not crowded, therefore,
it will follow that the target instance $i$ for this configuration is defined, 
and $i \leq \hat i$: which is sufficient to satisfy the second part of the definition of {\em good}.

To make this argument, we can treat the selection
of bins for each token in instance $\hat i$ as a balls in bins problem.
We therefore apply
Lemma~\ref{lem:balls} to $k$ and $k'=k_{\hat i}$,
which tells us that for any constant $\gamma \geq 9$,
the probability that instance $\hat i$ crowded is less than $1/N^{(\gamma/3) -2}$.
If we set out bin size constant $\gamma \geq 3c+9$, this probability is less than $1/N^{c+1}$.

Pulling together the pieces, for $\beta \geq c+3$ and $\gamma \geq 3c + 9$,
a union bound provides
that the probability that we fail to satisfy at least one of the two parts of the definition of {\em good}
is less than $2/N^{c+1} \leq 1/N^c$, satisfying the lemma statement. 
\end{proof}

Now that we have established that good configurations are likely,
we establish the below lemma about these configurations that follows directly from the definition of good and the mechanism
by which our algorithm updates estimates:

\begin{lemma}
In an execution with a good configuration with target instance $i$,
no node ever sets its local estimate to a value larger than $i$. That is, for all $u$ and all rounds, $est_u \leq i$.
\label{lem:good2}
\end{lemma}

We now continue our analysis by bounding the time required for all nodes to reach the target instance.
We do so with two arguments: the first concerning the rounds required for nodes to learn of a larger estimate
existing in the system, and the second concerning the rounds required for the largest estimate to increase
if it is still less than the target. 
For the following results, recall that $D$ is the network diameter.
\begin{lemma}
Fix an execution with a good configuration with target instance $i$.
Assume that at the beginning of round $r$ of this execution the largest estimate in the system is $i_{max} \leq i$.
By round $r' = r + {O}(D k_{i_{max}}\log^3{N})$ either: the largest estimate in the system is larger than $i_{max}$,
or all nodes have estimate $i_{max}$.
\label{lem:prop1}
\end{lemma}
\begin{proof}
Fix a node $u$ that has $est_u=i_{max}$ at the beginning of round $r$.
If $u$ maintains that estimate at the beginning of its next instance $i_{max}$ phase,
then during that phase it will advertise at least one $1$-bit (as it has at least its own tag
in one of the bins for this instance).
It follows that all $u$'s neighbors in the underlying topology will learn that $u$ has $est_u = i_{max}$
and will upgrade their estimate to $i_{max}$, if their estimate is currently less than this value.
We can then repeat this argument for $u$'s neighbors,
then their neighbors, and so on until either: at least one node adopts a larger estimate than $i_{max}$
(which might impede the application of this logic), or all nodes adopt $i_{max}$.
If the first event occurs, we satisfy the lemma statement.
If the first event does not occur,
the second event will occur after at most diameter $D+1$ instance $i_{max}$ phases (the extra phase
upper bounds the rounds required between round $r$ and the start of the next instance $i_{max}$ phase).
The number of rounds to complete an instance $i_{max}$ phase can be calculated
as: $k_{i_{max}}$ bins times $\gamma \log{N}$ blocks per bin times $\ell + \log{N} = O(\log{N})$ instance $i_{max}$ rounds
per block times $\log{N}$ real rounds for each instance $i_{max}$ rounds.
This product evaluates to $O(k_{i_{max}} \log^3{N})$ rounds per instance.
Therefore, $O(D k_{i_{max}} \log^3{N})$ rounds are sufficient to guarantee the lemma statement holds.
\end{proof}

\begin{lemma}
Fix an execution with a good configuration with target instance $i$.
Assume that at the beginning of round $r$ of this execution the largest estimate in the system is $i_{max} < i$.
By round $r' = r + {O}(D k_{i_{max}}\log^3{N})$ the largest estimate in the system is larger than $i_{max}$.
\label{lem:prop2}
\end{lemma}
\begin{proof}
We start by applying Lemma~\ref{lem:prop1} to $i_{max}$ and round $r$.
This establishes that by round $r' = r + {O}(D k_{i_{max}}\log^3{N})$ rounds either all nodes
have estimate $i_{max}$, or at least one node has an estimate larger than $i_{max}$.
If the latter is true than the lemma is satisfied directly at round $r'$.

Moving forward, therefore, assume all nodes have the same estimate $i_{max}$
by round $r'$. 
By assumption, $i_{max} < i$.
It follows that instance $i_{max}$ has at least one crowded bin.
Call this bin $j$. 
Let $T_j$ be the tags of the $\gamma \log{N}$ smallest tokens assigned to bin $j$ in instance $i_{max}$ in this configuration.
Because nodes spell out tags from order of smallest to largest, 
we know that any node that knows tags from $T_j$,
will assign each of these tags a block in any execution of instance $i_{max}$.

It follows, therefore, that in each execution of an $i_{max}$ phase,
if all nodes start that phase with an estimate of $i_{max}$,
then {\em each} of these tags in $T_j$ will spread another hop.
Applying the same argument as in the proof of Lemma~\ref{lem:prop1},
after at most $D$ executions of $i_{max}$ phases,
either at least one node has increased its estimate to a value larger than $i_{max}$,
or the tokens in $T_j$ will have spread to all nodes in the network.
If the latter event happens, then, by the definition of the algorithm,
all nodes will have discovered a crowded bin in instance $i_{max}$ and will increment their estimate.
Either way, the lemma is satisfied.
Therefore, by round $r' + {O}(D k_{i_{max}}\log^3{N}) = r+ {O}(D k_{i_{max}}\log^3{N})$,
the conditions of the lemma is satisfied---as required.
\end{proof}

The following key result leverages Lemmas~\ref{lem:prop1} and \ref{lem:prop2}
to bound the total rounds required for all nodes to permanently stabilize their estimates to the target instance.

\begin{lemma}
Fix an execution with a good configuration with target instance $i$.
By round $r = O(Dk_i\log^3{N})$, every node has estimate $i$. That is, for every node $u$, $est_u=i$ by round $r$.
\label{lem:target}
\end{lemma}
\begin{proof}
By the definition of our algorithm, estimates never decrease.
By Lemma~\ref{lem:good2}, no node will ever adopt an estimate greater than $i$.
Combined, it follows that
we can keep applying Lemma~\ref{lem:prop2} to increase
the largest estimate until the largest estimate reaches $i$.
We can then apply a single instance of Lemma~\ref{lem:prop1}
to ensure all nodes have this estimate---at which point the lemma will be permanently satisfied.

To bound the time required for these applications of the above lemmas,
we leverage our observation that the largest estimate can only increase.
It follows that in the worst case we apply Lemma~\ref{lem:prop2} exactly
once for each of the estimates leading up to the target $i$.
Because these estimates form a geometric sequence (e.g., $2,4,8,...$),
the total rounds needed for these applications of Lemma~\ref{lem:prop2} is upper bounded by:

\begin{eqnarray*}
{O}(D k_1\log^3{N}) +  {O}(D k_2\log^3{N}) + ... +  {O}(D k_i\log^3{N}) & =&  O\left( (D\log^3{N})(k_1 + k_2 + ... + k_i   ) \right)\\
& =& O(Dk_i\log^3{N})
\end{eqnarray*}

The final application of Lemma~\ref{lem:prop1} to spread estimate $i$ to all remaining nodes
once it exists in the system adds only a single aan additional $O(Dk_i\log^3{N})$ rounds.
The lemma statement follows.
\end{proof}

The preceding arguments bound the rounds required for useful information to propagate through
the network via the nodes' advertising bits.
We now conclude our proof by turning our attention to the rounds required for 
the actual tokens (which must be passed one at a time through pairwise connections) to spread.
We will tackle this problem by picking up where Lemma~\ref{lem:target} left off:
a point at which the system is prepared for the PPUSH instances executing
in the second half of blocks to make consistent progress.
We will apply our bound on PPUSH from Theorem~\ref{thm:ppush}
to establish the time required for this final propagation.
We will then leverage Theorem~\ref{thm:diam} to replace the network diameter
in our complexity with an upper bound expressed with respect to the network size and expansion.

\begin{theorem}
The CrowdedBin gossip algorithm solves the gossip problem in
$O((1/\alpha)k\log^6{N})$ rounds
when executed with tag length $b=1$ in a network with stability $\tau = \infty$.
\label{thm:crowded}
\end{theorem}
\begin{proof}
Assume for now that the configuration is good and $i$ its target instance.
Let round $r = O(Dk_i\log^3{N})$ be the round specified by Lemma~\ref{lem:target} for the network to converge its estimate.
That is, every node has the same estimate $i$ by round $r$.
By definition, no bin is crowded for instance $i$ in a good configuration.
It follows that {\em every} tag for {\em every} bin in this instance will be spread in  {\em every} round by the nodes that know that tag in that round.
Following the same propagation arguments used in Lemmas~\ref{lem:prop1} and~\ref{lem:prop2},
after at most $D$ more phases of instance $i$,
all nodes will know all tags.  This requires 
at most $O(Dk_i\log^3{N})$ rounds.
Therefore by some round $r' = O(Dk_i\log^3{N})$,
the system will have reached a {\em stable} state in which every node has the same estimate $i$
and knows the tag for every token in the system. This information will never again change so we can turn our 
attention for the rounds required to finish propagating the actual tokens after this point of stabilization.

To bound this token propagation time, fix an arbitrary token $t$ with tag $q$ in instance $i$.
Because we assume the system has stabilized,
every node has $q$ assigned to the same block of the same bin in their instance $i$ phase.
It follows that if we append together the last $\log{N}$ rounds from these blocks (i.e., the rounds in which
nodes run PPUSH for the tag described in the first $\ell$ rounds of the block),
we obtain a proper execution of PPUSH rumor spreading for token $t$ during these rounds.
That is, every time we come to the last $\log{N}$ rounds of $q$'s block,
all nodes are running PPUSH for rumor $t$,
picking up where they left off in the previous instance.

Applying Theorem~\ref{thm:ppush} from above,
it follows that with high probability in $N$,
$O(\log^4{N}/\alpha)$ rounds are sufficient for $t$ to spread to all nodes after stabilization.
Each phase provides $\log{N}$ rounds of PPUSH,
so $O(\log^3{N}/\alpha)$ phases are sufficient after stabilization. 

The key observation is that each execution of instance $i$ services {\em all} $k$ rumors after stabilization,
as each rumor has its own fixed bin in the instance $i$ phase.
Therefore, $O(\log^3{N}/\alpha)$ phases are sufficient to spread {\em all} $k$ rumors in parallel.
A union bound establishes that all $k \leq N$ instances succeed with a slightly reduced
high probability.

From a probability perspective,
we know from Lemma~\ref{lem:good} that the configuration is good with high probability.
We just argued above that if the configuration is good,
then with an additional high probability the tokens will all spread
in the stated time, once the system stabilizes.
We can increase both high probabilities to the desired exponent by increasing
the constant $\beta$ and $\gamma$ used in the definition of crowded bins,
and the constant factor in the time bound for PPUSH.
A union bound then shows that both good events occur with high probability.

From round cost perspective, we established that the time to stabilization is 
at most $O(Dk_i\log^3{N})$ rounds,
while the time to complete propagation after stabilization
is at most $O(\log^3{N}/\alpha)$ instance $i$ phases,
which each require $O(k_i\log^3{N})$ rounds.
The final time complexity is then
in: $O(Dk_i\log^3{N} + (k_i\log^6{N})/\alpha)$.

By the definition of a good configuration,
we know $k_i \leq 2k$,
and by Theorem~\ref{thm:diam},
we know $D=O(\log{N}/\alpha)$.
We can therefore simplify this complexity to $O((k\log^6{N})/\alpha)$ rounds, as required.
\end{proof}


\section{$\epsilon$-Gossip with $b=1$ and $\tau \geq 1$}
\label{sec:egossip}


In this section we consider $\epsilon$-Gossip: a relaxed version of the gossip problem that is parameterized
with some $\epsilon$, $0 < \epsilon < 1$ (e.g., as also studied in~\cite{dolev2007gossiping}).
In more detail,
the problem assumes all $n$ nodes start with a token.
To solve $\epsilon$-gossip there must be a subset $S$ of the $n$ nodes in the system,
where $|S| \geq \epsilon n$ {\em and} for every $u,v\in S$, $u$ knows $v$'s token
and $v$ knows $u$'s token.
Our goal here is to prove that for reasonably well-connected graphs and constant $\epsilon$,
{\em almost} solving gossip can be significantly faster than {\em fully} solving gossip.
In particular, we prove
that our SharedBit algorithm from before solves $\epsilon$-gossip in $O\left(\frac{n\sqrt{\Delta\log{\Delta}}}{(1-\epsilon)\alpha}\right)$
rounds. Given that $\Delta \leq n$, 
this is faster than the $O(n^2)$ required by SharedBit (for $k=n$) when $\epsilon$ is a constant fraction
and $\alpha = \omega(\log{\Delta}/(\sqrt{\Delta\log{\Delta}}))$.

\paragraph{Preliminaries.}
We restrict our attention in this analysis to the case where $\epsilon \geq 1/2$.
We can then handle smaller values for this fraction by applying the
below analysis for $\epsilon=1/2$: a value that (more than) solves
the problems for the smaller fraction, and at a cost of at most an extra constant factor
in the time complexity (i.e., when we replace $(1-\epsilon)$ in the denominator
with $(1-1/2)$, where $\epsilon < 1/2$ is the actual value we are analyzing,
the stated bound is less than a factor of two larger than what we would get with the smaller $\epsilon$).

A key tool in our analysis is a set that describes the frequency of different token sets
owned by nodes in the network at the beginning of a given round.
To do so, let $T$ be the set of tokens in the network.
The definition of $\epsilon$-gossip requires that $|T|=n$.
For each token subset $S\subseteq T$ and round $r \geq 1$, we define:

\[ count(S,r) = |\{ u\in V\mid T_u(r) = S \}|, \] 

\noindent where $T_u(r)$ is defined the same as in our above SharedBit analysis (i.e., the set of tokens $u$ knows at the beginning of round $r$).
Therefore, $count(S,r)$ equals the number of nodes with token set $S$ at the beginning of $r$.
We now use the definition of $count$ to define, for each round $r \geq 1$, the following multiset:

\[  F(r) = \{ (S, q) \mid (S\subseteq T)\wedge (q=count(S,r)) \wedge (q \geq 1)\} \]

\noindent This multiset contains all the token sets that appear at least once in the network at the beginning of round $r$,
 along with their frequency of occurrence.
Finally, we also make use of the following
potential function $\phi$, which was first defined in Section~\ref{sec:unstable:alg} to analyze SharedBit gossip:

\[ \forall r\geq 1: \phi(r) = \sum_{u\in V} \left(   n-|T_u(r)|   \right). \]


Our analysis will also leverage two useful lemmas from our earlier study of rumor spreading
in the mobile telephone model~\cite{ghaffari:2016}. The first lemma is graph theoretic,
and accordingly requires two definitions concerning graph properties. First, for a given graph $G=(V,E)$
and node set $S \subset V$, we define $B_G(S)$ to be the bipartite graph containing all (and only) the edges
from $E$ that connect a node in $S$ to a node in $V\setminus S$, with a vertex set consisting of these endpoints.
Second, for a given graph $H$, let $\nu(H)$ the {\em edge independence number} of $H$, which describes the size of
a maximum matching on $H$. We now proceed with our lemma:

\begin{lemma}[Adapted from~\cite{ghaffari:2016}]
\label{lem:msize}
Fix a graph $G=(V,E)$ with $|V| = n$ and vertex expansion $\alpha$.
Fix some $S\subset V$ such that $|S| \leq n/2$.
It follows that $\nu(B_G(S)) \geq |S| \cdot (\alpha/4).$
\end{lemma}

The second lemma adapted from~\cite{ghaffari:2016} is algorithmic in that bounds the performance
of a simple randomized strategy for approximating a maximum matching in a bipartite graph:

\begin{lemma}[Adapted from~\cite{ghaffari:2016}]
\label{lem:7.3}
Fix a network topology graph $G=(V,E)$ with maximum degree $\Delta$.  Fix some subset $C\subset V$.
Assume there is a matching $M$ of size $m\geq 1$ defined over $B_G(C)$.
Assume each node in $C$ randomly chooses a neighbor in $B_G(C)$ to send a connection proposal.
With constant probability, at least $\Omega(\frac{m}{\sqrt{\Delta \log{\Delta}}})$ nodes from $V\setminus C$ 
that are endpoints in $M$ will
receive a connection proposal from a node in $C$. 
\end{lemma}

\paragraph{Analysis.}
Our main strategy is to attempt to identify for each round a {\em coalition} of
nodes such that: (1) the size of the coalition is within a target range
$(\epsilon/2) n$ to $\epsilon n$; and (2)  no node in the coalition has the same token
set as a node outside the coalition.
If we can find such a coalition,
the graph property result captured in Lemma~\ref{lem:msize}
tells us that there are many edges between coalition and non-coalition nodes (where the definition
of ``many" depends on $\alpha$ and $\epsilon$). We can then show that a reasonable fraction
of these edges will connect and therefore reduce $\phi$. 
We begin this argument by leveraging the above definitions
to prove that either we can find such a coalition or we have already solved the problem.

\begin{lemma}
\label{lem:col}
Fix a round $r \geq 1$. One of the following must be true about this round:
(1) $\epsilon$-gossip is solved by the beginning of round $r$; or
(2) there exists a $C \subset F(r)$ such that:

\[  (\epsilon/2)n \leq    \sum_{(S,q)\in C} q \leq n\epsilon. \]
\end{lemma}
\begin{proof}
Let $q_{max} = \max\{ q : (*,q)\in F(r)\}$ (i.e., the number of nodes that own the set owned by the most nodes in $r$).
We consider three cases for $q_{max}$ and show that all three satisfy our lemma.

The first case is that $q_{max} > n\epsilon$.
In this case, we have identified a token set $S$ that is owned by more than $n\epsilon$ nodes.
Let $V_S$ be the set of nodes that own $S$ at the beginning of $r$.
Because every node starts with its own token in its token set,
and no token ever leaves a token set,
we know for each $u\in V_S$, $u$'s token is in $S$.
It follows that every node in $V_S$ knows the token of every other node in this set---meaning we have solved $\epsilon$-gossip
and therefore satisfy option (1) from the lemma statement.

The second case is that  $(\epsilon/2)n \leq     q_{max} \leq n\epsilon$.
In this case, we can set $C=\{  (S,q_{max}) \}$,
where $S$ is the set we identified owned by $q_{max}$ nodes (if more than $1$, choose one arbitrarily),
 and directly satisfy option (2) from the lemma statement.

The third and final case is that $q_{max} < (\epsilon/2)n$.
In this case, we can apply the following simple greedy strategy for defining $C$:
keep adding pairs from $F(r)$ to $C$ in {\em decreasing} order of  $q$ values
until $\sum_{(S,q)\in C} q$ first grows larger than $(\epsilon/2)n$.
By our case assumption, every $q$ value in $F$ is less than $(\epsilon/2)n$.
Therefore, the step of the greedy strategy that first pushes us over the $(\epsilon/2)n$ threshold
must increase this sum to fall within our target range of $(\epsilon/2)n$ and $\epsilon n$.
That is, the greedy strategy described above will always terminate having
identified a set $C$ that satisfies option (2) from the lemma statement.
\end{proof}

Repeatedly applying Lemma~\ref{lem:col} will provide that in each round either we are done
with the $\epsilon$-gossip problem or we have a large coalition that is likely to generate lots of progress
toward solving the problem.
We are now ready to pull together our pieces to prove our main theorem.
The main technical contribution of the below proof is arguing that a large coalition likely generates
lots of new token transfers. 
This claim will pull from Lemmas~\ref{lem:7.3} and~\ref{lem:msize} from above,
as well as Lemma~\ref{lem:unstable:adv} from the SharedBit analysis in Section~\ref{sec:unstable:alg}.

\begin{theorem}
\label{thm:egossip}
Fix some $\epsilon$, $0 < \epsilon < 1$.
The SharedBit gossip algorithm solves the $\epsilon$-gossip problem in $O\left(\frac{n\sqrt{\Delta\log{\Delta}}}{(1-\epsilon)\alpha}\right)$ rounds
when executed with shared randomness with tag length $b=1$ in a network with stability $\tau \geq 1$.
\end{theorem}
\begin{proof}
Fix some $\epsilon$ that satisfies the theorem statement.
Assume w.l.o.g.~that $\epsilon \geq 1/2$ (as argued at the beginning of this 
analysis, if $\epsilon$ is smaller, we can apply our analysis for $\epsilon=1/2$ which
more than solves the problem at the cost of only an extra constant factor in the stated time complexity).
We begin by focusing on a single round, then extend the argument to the full execution.
In particular, fix a round $r$, $1\leq r \leq cN^2$ (i.e., a round for which we still have 
bits in the shared string $\hat r$ used by SharedBit). Let $G_r = (V,E)$ be the network topology graph in this round.
Assume $\epsilon$-gossip has not finished by the beginning of this round.
By Lemma~\ref{lem:col}, there exists a $C\subset F(r)$ such that:
\[  (\epsilon/2)n \leq    \sum_{(S,q)\in C} q \leq n\epsilon. \]
Let $V_C$ be the set of nodes that start round $r$ with one of the token sets in $C$.
By our above assumption: $(\epsilon/2)n \leq   |V_C|  \leq n\epsilon$.

Let $q = \min\{ |V_C|, |V \setminus V_C|\}$.
It follows that $q \leq n/2$.
By Lemma~\ref{lem:msize}, therefore,
there exists a matching $M$ of size $m\geq (\alpha/4)q$ in $B_{G_r}(V_C)$
(the bipartite subgraph of $G_r$ that keeps only edges from $E$ with one endpoint in $V_C$ and one endpoint in $V\setminus V_C$).
For each edge $e\in M$,
we define $e.c$ to be the endpoint from $e$ in $V_C$
and $e.v$ to be the endpoint from $e$ in $V \setminus V_C$. 
We say an edge $e\in M$ is {\em wasted}
if both endpoints in $e$ advertise the same bit; i.e., $b_{e.c}(r) = b_{e.v}(r)$.
By the definition of the coalition used in Lemma~\ref{lem:col},
it follows for each $e\in M$ it must be the case that $T_{e.c}(r) \neq T_{e.v}(r)$.
We can therefore apply Lemma~\ref{lem:unstable:adv} which provides
that the probability they advertise different bits is $1/2$. 
The probability that $e$ is wasted is therefore also $1/2$.

To argue more precisely about wasted edges we define some random variables.
For each $e\in M$,
let $X_e$ be the random indicator variable that evaluates to $1$ if $e$ is wasted and otherwise evaluates to $0$.
Let $Y= \sum_{e\in M} X_e$.
By linearity of expectation and our above argument about the probability of wastefulness, it follows:
$E(Y) = m/2$.

We now want to bound the probability that the actual number of wasted edges is not too much larger than $E(Y)$.
We cannot apply a Chernoff-style bound as there might be dependency between the outcomes of different edges
in $M$ (as they may share tokens, and therefore share random bits used to determine their tag).
To sidestep these issues,
we apply Markov's Inequality (Theorem~\ref{thm:markov} in Section~\ref{sec:model}) to derive
the following:

\[  \Pr\left(  Y \geq (3/2)\cdot E(Y)\right) \leq \frac{E(Y)}{(3/2)\cdot E(Y)} = 2/3. \]

Notice that $(3/2)\cdot E(Y) = (3/4)\cdot m$.
We can therefore reword this result to say that with probability at least $1/3$,
at least $m/4$ edges in $M$ are not wasted.
For clarity, we will subsequently refer to an edge from $M$ that is {\em not wasted} as an edge that is {\em primed}
(as in the edge is primed for the possibility of its endpoints connecting in a manner that helps spread tokens).

Moving forward in this analysis,
assume this event occurs, and therefore at least $m/4$ edges in $M$ are primed.
Let $\hat M\subseteq M$ be this set of primed edges.
(Notice, because $m\geq 1$ and the size of $\hat M$ must be a whole number,
we know $\hat M$ is non-empty under this assumption.)

We want to now apply Lemma~\ref{lem:7.3} to the connections described by $\hat M$.
To do so, let $\hat C$ be the endpoints in $\hat M$ that advertise a $1$ in this round.
Let $\hat G$ be the topology graph $G_r$ for this round {\em modified} such that we remove every
node that is not in $\hat C$, but neighbors $\hat C$ and also advertises a $1$ (along with their incident edges). 
We emphasize two properties of this modification: (1) by definition, no node in $\hat M$ is removed by this step;
(2) it is correct to say that nodes in $\hat C$ will choose a neighbor from $\hat G$ uniformly
to send a connection proposal, because the SharedBit algorithm only has nodes that advertise a $1$ choose
among neighbors that advertise a $0$, and we only removed neighbors from $\hat C$ nodes that also advertised a $1$.

We can therefore apply Lemma~\ref{lem:7.3} with $G=\hat G$, $C=\hat C$, and $M=\hat M$.
It follows that with constant probability,
at least $\Omega(|\hat M|/\sqrt{\Delta \log{\Delta}  }) = \Omega(m\sqrt{\Delta \log{\Delta}  })$ nodes in $\hat M$ receive a connection proposal from 
their neighbor in this matching.
Each such node $u$ will subsequently connect with {\em some} node $v$ in this round (though not necessarily its neighbor in $\hat M$).
By Lemma~\ref{lem:unstable:adv}, however, $T_{u}(r) \neq T_{v}(r)$ (as each advertised different bits in $r$),
so each of these connections reduces $\phi$ by at least $1$.

Combining our probabilistic events from above, it follows that with constant probability,
$\phi(r+1) - \phi(r) \geq \delta \in \Omega(\frac{\alpha q}{\sqrt{\Delta \log{\Delta}  }}    )$,
where, as defined above, $q=\min\{|V_C|, |V \setminus V_C|\}$.
Let us call a round in which this event occurs a {\em good} round.
To bound the number of good rounds until $\phi$ reduces to $0$ (and the $\epsilon$-gossip problem is solved, regardless of $\epsilon$),
we must first lower bound the size of $\delta$.
To do so, we first note that $|V \setminus V_C| \geq (1-\epsilon)n$.
It follows that in the case where $q=|V \setminus V_C|$, we know $q \geq (1-\epsilon)n$.
On the other hand, if $q= |V_C|$,
we can apply our assumption that $\epsilon \geq 1/2$ (see the beginning of this proof)
to conclude that $q \geq (1/3)(1-\epsilon)n$.
Combined: $(1/3)(1-\epsilon)(n)$ provides a general lower bound on $q$ for all rounds.

We now know that in a good round $r$:

\[ \phi(r+1) - \phi(r) \geq \delta \in \Omega\left(\frac{\alpha (1-\epsilon)n}{\sqrt{\Delta \log{\Delta}  }}    \right).\]

Because $\phi(1) \leq n^2$ and $\phi$ can only decrease,
it follows that 

\[ n^2/\delta = t_{good} \in O\left(\frac{n \sqrt{\Delta \log{\Delta}  }}{ \alpha(1-\epsilon)  } \right)\]

\noindent good rounds are sufficient to conclude gossip. As established above, 
the probability of a given round being good is lower bounded by a constant,
regardless of the execution history preceding that round.
For each round $r$, let $X_r$ be the random indicator variable that evaluates
to $1$ if and only if $r$ is good. We know $Pr(X_r = 1) \geq p$, for the constant probability mentioned above.
Therefore, in expectation, $t_{good}/p \in \Theta(t_{good})$ rounds are sufficient to achieve $t_{good}$ good rounds.
To obtain a high probability result we cannot directly apply a Chernoff bound to these indicator
variables as they are not necessarily independent. 
Each $X_r$, however, stochastically dominates the trivial random variable $\hat X_r$ that evaluates to $1$
with probability $p$. We can then apply a concentration result to the expectation calculated
on the $\hat X$ variables to determine that $\Theta(t_{good})$ rounds are sufficient, with high probability in $n$.

Pulling together the pieces,
by Lemma~\ref{lem:col},
for each round $r$, either
we have solved $\epsilon$-gossip or we can find a coalition that provides us a constant probability of $r$
being a good round.
With high probability, the latter can occur at most $O(t_{good})= O\left(\frac{n\sqrt{\Delta\log{\Delta}}}{(1-\epsilon)\alpha}\right)$ times before
we still solve the problem.
\end{proof}

The following corollary follows directly from our analysis in 
Section~\ref{sec:unstable:derandomize} concerning the elimination of the shared randomness
assumption when solving gossip with SharedBit.

\begin{corollary}
Fix some $\epsilon$, $0< \epsilon < 1$.
There exists a bit string multiset ${\cal R'}$,
such that
the SimSharedBit gossip algorithm using this ${\cal R'}$
solves the $\epsilon$-gossip problem in $O\big(\frac{n\sqrt{\Delta\log{\Delta}}}{(1-\epsilon)\alpha} + (1/\alpha)\Delta^{1/\tau}\log^6{N}\big) 
= \tilde{O}\big( \frac{n\sqrt{\Delta\log{\Delta}}}{(1-\epsilon)\alpha} \big)$ rounds when
executed with tag length $b=1$ in a network with stability $\tau \geq 1$.
\end{corollary}

\bibliographystyle{plain}
\bibliography{p2p}

\begin{thebibliography}{10}

\bibitem{firechat}
{FireChat Phone-to-Phone App}.
\newblock \url{ http://www.opengarden.com/FireChat }.

\bibitem{smartphones}
{Latest mobile statistics: key figures (Ericsson Mobility Report)}.
\newblock
  \url{https://www.ericsson.com/mobility-report/latest-mobile-statistics}.

\bibitem{burleigh2003delay}
Scott Burleigh, Adrian Hooke, Leigh Torgerson, Kevin Fall, Vint Cerf, Bob
  Durst, Keith Scott, and Howard Weiss.
\newblock Delay-tolerant networking: an approach to interplanetary internet.
\newblock {\em IEEE Communications Magazine}, 41(6):128--136, 2003.

\bibitem{camps2013device}
Daniel Camps-Mur, Andres Garcia-Saavedra, and Pablo Serrano.
\newblock Device-to-device communications with wi-fi direct: overview and
  experimentation.
\newblock {\em IEEE wireless communications}, 20(3):96--104, 2013.

\bibitem{chierichetti2010rumour}
Flavio Chierichetti, Silvio Lattanzi, and Alessandro Panconesi.
\newblock Rumour spreading and graph conductance.
\newblock In {\em Proceedings of the ACM-SIAM symposium on Discrete Algorithms
  (SODA)}, 2010.

\bibitem{kuhn:bounded}
Sebastian Daum, Fabian Kuhn, and Yannic Maus.
\newblock Rumor spreading with bounded in-degree.
\newblock In {\em International Colloquium on Structural Information and
  Communication Complexity (SIRROCO)}, 2016.

\bibitem{dolev2007gossiping}
Shlomi Dolev, Seth Gilbert, Rachid Guerraoui, and Calvin Newport.
\newblock Gossiping in a multi-channel radio network.
\newblock In {\em Proceedings of the Symposium on Distributed Computing
  (DISC)}, 2007.

\bibitem{fountoulakis2010rumor}
Nikolaos Fountoulakis and Konstantinos Panagiotou.
\newblock Rumor spreading on random regular graphs and expanders.
\newblock In {\em Approximation, Randomization, and Combinatorial Optimization.
  Algorithms and Techniques}, pages 560--573. Springer, 2010.

\bibitem{frieze1985shortest}
Alan~M Frieze and Geoffrey~R Grimmett.
\newblock The shortest-path problem for graphs with random arc-lengths.
\newblock {\em Discrete Applied Mathematics}, 10(1):57--77, 1985.

\bibitem{telephone1}
Alan~M Frieze and Geoffrey~R Grimmett.
\newblock The shortest-path problem for graphs with random arc-lengths.
\newblock {\em Discrete Applied Mathematics}, 10(1):57--77, 1985.

\bibitem{ghaffari:2016}
Mohsen Ghaffari and Calvin Newport.
\newblock How to discreetly spread a rumor in a crowd.
\newblock In {\em Proceedings of the International Symposium on Distributed
  Computing (DISC)}, 2016.

\bibitem{telephone2}
George Giakkoupis.
\newblock Tight bounds for rumor spreading in graphs of a given conductance.
\newblock In {\em Proceedings of the Symposium on Theoretical Aspects of
  Computer Science (STACS)}, 2011.

\bibitem{giakkoupis2011tight}
George Giakkoupis.
\newblock Tight bounds for rumor spreading in graphs of a given conductance.
\newblock In {\em Proceedings of the Symposium on Theoretical Aspects of
  Computer Science (STACS)}, 2011.

\bibitem{giakkoupis2014tight}
George Giakkoupis.
\newblock Tight bounds for rumor spreading with vertex expansion.
\newblock In {\em Proceedings of the ACM-SIAM Symposium on Discrete Algorithms
  (SODA)}, 2014.

\bibitem{telephone3}
George Giakkoupis and Thomas Sauerwald.
\newblock Rumor spreading and vertex expansion.
\newblock In {\em Proceedings of the ACM-SIAM symposium on Discrete Algorithms
  (SODA)}, pages 1623--1641, 2012.

\bibitem{giakkoupis2012rumor}
George Giakkoupis and Thomas Sauerwald.
\newblock Rumor spreading and vertex expansion.
\newblock In {\em Proceedings of the ACM-SIAM symposium on Discrete Algorithms
  (SODA)}, pages 1623--1641. SIAM, 2012.

\bibitem{gomez2012overview}
Carles Gomez, Joaquim Oller, and Josep Paradells.
\newblock Overview and evaluation of bluetooth low energy: An emerging
  low-power wireless technology.
\newblock {\em Sensors}, 12(9):11734--11753, 2012.

\bibitem{gopalsamy2002reliable}
Thiagaraja Gopalsamy, Mukesh Singhal, D~Panda, and P~Sadayappan.
\newblock A reliable multicast algorithm for mobile ad hoc networks.
\newblock In {\em Proceedings of the IEEE International Conference on
  Distributed Computing Systems (ICDCS)}, pages 563--570. IEEE, 2002.

\bibitem{kuhn2010distributed}
Fabian Kuhn, Nancy Lynch, and Rotem Oshman.
\newblock Distributed computation in dynamic networks.
\newblock In {\em Proceedings of the Symposium on Principles of Distributed
  Computing (PODC)}, pages 513--522. ACM, 2010.

\bibitem{mark2015peer}
David Mark, Jayant Varma, Jeff LaMarche, Alex Horovitz, and Kevin Kim.
\newblock Peer-to-peer using multipeer connectivity.
\newblock In {\em More iPhone Development with Swift}, pages 239--280.
  Springer, 2015.

\bibitem{newmans}
Ilan Newman.
\newblock Private vs. common random bits in communication complexity.
\newblock {\em Information processing letters}, 39(2):67--71, 1991.

\bibitem{newport:2017}
Calvin Newport.
\newblock Leader election in a smartphone peer-to-peer network.
\newblock In {\em Proceedings of the IEEE International Parallel and
  Distributed Processing Symposium (IPDPS)}, 2017.
\newblock Full version available online at:
  \url{http://people.cs.georgetown.edu/~cnewport/pubs/le-IPDPS2017.pdf}.

\bibitem{shah:2009}
Devavrat Shah et~al.
\newblock Gossip algorithms.
\newblock {\em Foundations and Trends in Networking}, 3(1):1--125, 2009.

\end{thebibliography}




\end{document}